\documentclass[10pt,journal,compsoc]{IEEEtran}
\usepackage{amsmath,amssymb,amsfonts}
\usepackage{amsthm}
\usepackage{algorithm}
\usepackage{tablefootnote}
\usepackage{algpseudocode}
\newcommand{\fade}[1]{\textcolor{gray}{#1}}
\usepackage{url}
\usepackage{graphicx}
\usepackage{textcomp}
\usepackage{booktabs} 
\usepackage{xcolor}
\usepackage{color}
\usepackage{footnote}

\usepackage[hidelinks]{hyperref}
\usepackage{subcaption}
\usepackage{bm}
\usepackage{multirow}
\newlength\myindent
\newtheorem{definition}{Definition}
\newtheorem{theorem}{Theorem}
\newcommand*\fedavg{\textsc{FedAvg-SIA}}
\newcommand*\fedsgd{\textsc{FedSGD-SIA}}
\newcommand*\fedmd{\textsc{FedMD-SIA}}

\usepackage[framemethod=TikZ]{mdframed}
\usepackage{threeparttable}
\usepackage{tikz}
\newcommand{\pie}[1]{%
\begin{tikzpicture}
 \draw (0ex,0ex) circle (1ex);
 \fill (0ex,-1ex) arc (-90:(#1-90):1ex) -- (0ex,-1ex) -- cycle;
\end{tikzpicture}%
}

\ifCLASSOPTIONcompsoc
  \usepackage[nocompress]{cite}
\else
  \usepackage{cite}
\fi

\ifCLASSINFOpdf
\else
\fi
\begin{document}
\title{Source Inference Attacks: Beyond Membership Inference Attacks in Federated Learning}
\author{\IEEEauthorblockN{Hongsheng Hu\IEEEauthorrefmark{1},
        Xuyun Zhang\IEEEauthorrefmark{2}\thanks{\IEEEauthorrefmark{2}Xuyun Zhang is the corresponding author.},
        Zoran Salcic\IEEEauthorrefmark{1},~\IEEEmembership{Life Senior Member, IEEE,}
        Lichao Sun\IEEEauthorrefmark{3},\\
        Kim-Kwang Raymond Choo\IEEEauthorrefmark{4},~\IEEEmembership{Senior Member, IEEE,}
            and Gillian Dobbie\IEEEauthorrefmark{1}}
            \\
\IEEEauthorblockA{\IEEEauthorrefmark{1}The University of Auckland, New Zealand}

\IEEEauthorblockA{\IEEEauthorrefmark{2}Macquarie University, Australia}

\IEEEauthorblockA{\IEEEauthorrefmark{3}Lehigh University, USA} 

\IEEEauthorblockA{\IEEEauthorrefmark{4}The University of Texas at San Antonio, USA}
}


\markboth{}%
{}

\IEEEtitleabstractindextext{%
\begin{abstract}
Federated learning (FL) is a popular approach to facilitate privacy-aware machine learning since it allows multiple clients to collaboratively train a global model without granting others access to their private data. It is, however, known that FL can be vulnerable to \textit{membership inference attacks} (MIAs), where the training records of the global model can be distinguished from the testing records. Surprisingly, research focusing on the investigation of the source inference problem appears to be lacking. We also observe that identifying a training record's source client can result in privacy breaches extending beyond MIAs. For example, consider an FL application where multiple hospitals jointly train a COVID-19 diagnosis model, membership inference attackers can identify the medical records that have been used for training, and any additional identification of the source hospital can result the patient from the particular hospital more prone to discrimination. Seeking to contribute to the literature gap, we take the first step to investigate source privacy in FL. Specifically, we propose a new inference attack (hereafter referred to as \textit{source inference attack} -- SIA), designed to facilitate an honest-but-curious server to identify the training record's source client. The proposed SIAs leverage the Bayesian theorem to allow the server to implement the attack in a non-intrusive manner without deviating from the defined FL protocol. We then evaluate SIAs in three different FL frameworks to show that in existing FL frameworks, the clients sharing gradients, model parameters, or predictions on a public dataset will leak such source information to the server. We also conduct extensive experiments on various datasets to investigate the key factors in an SIA. The experimental results validate the efficacy of the proposed SIAs, e.g., an attack success rate of 67.1\% (baseline 10\%) can be achieved when the clients share model parameters with the server. Comprehensive ablation studies demonstrate that the success of an SIA is directly related to the overfitting of the local models.
\end{abstract}

\begin{IEEEkeywords}
Federated Learning, Membership Inference Attacks, Source Inference Attacks, Privacy Leakage.
\end{IEEEkeywords}}

\maketitle

\IEEEdisplaynontitleabstractindextext
% \fi
\IEEEpeerreviewmaketitle

\section{Introduction}\label{sec::01}
Recent deep learning advances have partly contributed to the building of powerful machine learning (ML) models from large datasets. In practice, however, data often resides across different organizational entities (also referred to as data islands). The data records of a single entity, perhaps with the exception of extremely large technology organizations, are generally limited and do not represent the entire data distribution. Thus, stakeholders (e.g., consumers) can generally benefit if different data owners can collaboratively train a joint machine learning model based on the union of different datasets. For example, our society will benefit if different countries and organizations can collaborate to collaboratively train COVID-19 diagnosis models using the broad range of medical data records in these different entities. The exacting privacy regulations (e.g., GDPR \cite{mantelero2013eu} in the European Union and CCPA \cite{de2018guide} in the United States), however, complicate such collaborative efforts. 

Federated learning (FL) is one approach that can be utilized to circumvent the limitations due to data islands, by allowing multiple clients coordinated by a central server to train a joint ML model in an iterative manner~\cite{mcmahan2017communication,smith2017federated,caldas2018leaf,bonawitz2019towards}. In FL, the clients send their model updates to the server but never their raw training dataset, thereby leading to a privacy-aware paradigm for collaborative model training. For the example mentioned above, FL can greatly facilitate the hospitals wishing to train a joint COVID-19 diagnosis model from the distributed data in different hospitals. A real-life case in~\cite{xu2020collaborative} has shown the successful adoption of FL where an ML model for COVID-19 diagnosis has been trained with the usage of the geographically distributed chest computed tomography data collected from different patients at different hospitals.

However, many recent studies~\cite{melis2019exploiting,nasr2019comprehensive,zhu2019deep,wang2019beyond,truex2019hybrid,zhou2022ppa} have shown that FL does not provide sufficient privacy
guarantees, because sensitive information from the training data can still be revealed during the communication process. In FL, the clients transmit necessary information from updates, {e.g.}, gradients, to the central server for global model training. Because the updates are derived from the clients' private training data, there are several recently proposed privacy attacks trying to infer the privacy of the clients from such updates, such as data reconstruction attacks~\cite{hitaj2017deep}, property inference attacks~\cite{ganju2018property}, preference profiling attacks~\cite{zhou2022ppa}, and membership inference attacks (MIAs)~\cite{shokri2017membership}. Among such attacks, MIAs aim to identify whether or not a data record was in the training dataset of a target model ({i.e.,} a member). While an MIA seems like a simple attack, it can impose severe privacy risks in many settings where knowing that someone was in a dataset is harmful~\cite{homer2008resolving}. For example, by identifying the fact that a clinical record was in the training dataset of a medical model associated with a certain disease, MIAs enable an attacker to know that the owner of the clinical record has a high chance of having the disease.

{In FL, the training data of the FL system consist of all the training records in the datasets of the clients because they collaboratively build a global federated model. Thus, MIAs of existing research in the context of FL \cite{melis2019exploiting,nasr2019comprehensive,pustozerova2020information,hu2022membership} are designed to distinguish the training records from the testing records of the global model, i.e., they do not require to identify which client owns a training record, {i.e.,} the source of the training record. However, it is important and practical to explore source privacy in FL beyond membership privacy because the leakage of the source information can further breach privacy. For example, in the previously mentioned FL application where multiple hospitals jointly train a COVID-19 diagnosis model, MIAs can only tell who has had a COVID-19 test, but the further identification of the source hospital of the people could make them more prone to discrimination, especially when the hospital is located in a high-risk region or country~\cite{devakumar2020racism}. Other types of attacks in FL such as property inference attacks also fail to explore the source privacy of clients because they are either designed to infer other types of private information or the private information inferred does not attribute to a specific client. For example, property inference attacks in FL \cite{melis2019exploiting} can infer a certain property in the training data but can not attribute the property to the client who owns this property.}

{This paper proposes a novel inference attack, named \textit{Source Inference Attack} (SIA), that determines which client in FL owns a training record. The SIA can be considered a stronger attack based on the foundation of MIAs, {i.e.}, after identifying which data records are training records via MIAs, the attacker further implements SIAs to determine which client a training record comes from. For practical reasons, we consider the server can be a semi-honest (honest-but-curious) attacker who passively tries to learn the secret based on the information on the identities of clients and communications between them. Specifically, the semi-honest server means that the server tries to infer the private information of the clients without interfering with the federated training, {i.e.,} without deviating from the defined FL protocol, which is the main challenge of SIAs. Note that the attacker can be one of the clients, but we argue that it is impractical in this case for SIAs because the client does not know the identities of the other clients, and it can only access the global model via communications to the server \cite{lyu2022privacy}. }

We leverage the Bayesian theorem to analyze how to effectively implement SIAs in FL. We demonstrate that an honest-but-curious server can estimate the source of a training record in an SIA by leveraging the prediction loss of the local models. More specifically, we theoretically demonstrate that the client with the smallest prediction loss on the training record should have the highest probability of owning it. To demonstrate the feasibility of the proposed SIAs to different FL frameworks, we propose three \textsc{FL-Sia} frameworks that enable the server to conduct the SIAs in three FL frameworks, FedSGD~\cite{mcmahan2017communication}, FedAvg~\cite{mcmahan2017communication}, and FedMD~\cite{li2019fedmd}. The purpose of selecting the three FL frameworks is to show that in existing FL frameworks, the clients sharing gradients, model parameters, or predictions on a public dataset will lead to source information leakage to the server. We conduct extensive experiments on six datasets and different model architectures under various FL settings to evaluate the effectiveness of SIAs. The experiment results validate the efficacy of our proposed SIAs. We conduct a detailed ablation study to investigate how data distributions across the clients and the number of local epochs in FL affect the performance of an SIA. An important finding is that the success of an SIA is directly related to the overfitting of the local models, which is mainly caused by the non-IID data distribution across the clients. 

The main contributions of this paper are three-folds: 
\begin{itemize}
    \item We propose a novel inference attack in federated learning (FL), named as source inference attack (SIA), which infers the source client of a training record. Beyond membership inference attacks, SIAs can further breach the privacy of the training records in FL. 
    \item We innovatively adopt the Bayesian theorem to analyze how an honest-but-curious server can implement SIAs in a non-intrusive manner to infer the source of a training record with the highest probability by using the prediction loss of local models. 
    \item We show the feasibility and effectiveness of SIAs in three FL frameworks, including FedSGD, FedAvg, and FedMD. We perform extensive experiments to empirically evaluate SIAs in the three frameworks with various datasets and under different FL settings. The results validate the efficacy of the proposed SIA. Our proposed SIAs shed new light on how FL reveals sensitive information and the need to build more private FL frameworks. 
\end{itemize}

{This work extends our earlier conference paper \cite{hu2021source}. In Section \ref{sec::02}, we introduce additional background material to give the readers have a broader understanding of the role and impact of SIAs in FL. In Section \ref{sec::03}, we extend the earlier work  in the following ways:}
\begin{itemize}
    \item {We use a threat model to describe the attack goal, target FL systems, attacker, and attack knowledge. }
    \item {We provide systematic theoretical analysis and propose theorems to show how to leverage the prediction loss of the models for conducting SIAs in FL. For the corresponding theorems, we also present detailed corresponding proofs.}
    \item {We show how to implement SIAs in two other commonly-used FL frameworks, FedSGD and FedMD, while in \cite{hu2021source} we only introduced the implementation of SIAs in FedAvg. This extension demonstrates the broader applicability of SIAs in FL.}
\end{itemize}
    
{In Section \ref{sec::04}, we describe our comprehensive experimental evaluation, including the addition in Section \ref{sec::why} to explain why the proposed SIAs can work in FL. In Section \ref{sec::05}, we provide more comprehensive experiments to investigate whether the popular defense method of differential privacy can mitigate SIAs. Moreover, in Section \ref{sec::05}, we discuss the limitations and potential research opportunities of our proposed SIAs. In Section \ref{sec::06}, we also include additional related work to give the readers a broader picture of privacy attacks and defenses in FL. Source code for implementing SIAs in FedSGD, FedAvg, and FedMD are also included\footnote{\url{https://github.com/HongshengHu/SIAs-Beyond_MIAs_in_Federated_Learning}}, while in \cite{hu2021source} we only provide code for SIAs in FedAvg.}

\section{Preliminaries}\label{sec::02}
This section reviews the background of federated learning and membership inference attacks.

\subsection{Federated Learning}\label{sec::pre_fl}
Because data usually exists in the form of isolated islands and central storage is impractical due to privacy regulations and laws, federated learning (FL) has been proposed to allow multiple clients to collaboratively train a machine learning model in an interactive manner. During the training phase of FL, the clients send necessary information of the updates but never their private datasets to the central server. Because the updates contain less information than the raw training data, FL has obvious privacy advantages compared to data center training~\cite{mcmahan2017communication}.

\noindent \textbf{Horizontal and vertical FL.} Based on the feature space or the sample ID space the local datasets share, FL can be divided into horizontal FL and vertical FL \cite{yang2019federated,li2021survey}. Horizontal FL, aka. cross-device FL, describes FL scenarios where the local datasets share the same feature space but are different in samples. An example of horizontal FL is multiple regional banks having different users from their respective regions, while the feature spaces of such users are the same because the banks have very similar businesses \cite{cheng2020federated}. Vertical FL, aka. cross-silo FL, describes FL scenarios where the local datasets share the same or similar sample ID space but differ in feature space. An example of vertical FL is two different commercial companies in the same city having the same or very similar customers in the area. However, due to different business modes, the commercial company of the bank has the user’s revenue and expenditure transactions, while the commercial company of e-commerce owns the user’s browsing and purchasing history. Vertical FL enables the two different companies to jointly build a model for predicting users' living behaviors \cite{yang2019federated}.

\noindent \textbf{Homogeneous and heterogeneous FL.} Based on architectures, FL frameworks can be divided into two categories, {i.e.,} FL with a homogeneous architecture and FL with a heterogeneous architecture~\cite{lyu2020threats}. In FL with a homogeneous architecture, the local models have the same architecture as the global model, and there are two forms of FL~\cite{mcmahan2017communication}: i) FedSGD, in which each of the clients sends gradients calculated on its local data to the server; ii) FedAvg, in which each of the clients sends the calculated local models' parameters to the server. FedSGD has the advantage of convergence guarantees of the global model but requires frequent communication between the clients and the server~\cite{yin2021comprehensive}. FedAvg is more communication efficient but the global model may not perform well when the training data across the clients are highly non-identically distributed~\cite{li2020fedprox}. In FL with a heterogeneous architecture, the local models do not have to share the same architecture as the global model, while each client can still benefit during the federated training process. FedMD~\cite{li2019fedmd} (Federated Model Distillation) is a novel FL framework with a heterogeneous architecture, which shares the knowledge of each client's local model via their predictions on an unlabeled public dataset instead of the local model's parameters. Compared to FedAvg, FedMD eliminates the risk of white-box inference attacks~\cite{melis2019exploiting} and has the advantage of reduced communication costs. 

Note that there are many other FL frameworks such as FedProx~\cite{li2020fedprox}, SCAFFOLD \cite{karimireddy2020scaffold}, FedDF~\cite{lin2020ensemble}, and Cronus~\cite{chang2019cronus} that are proposed to solve different challenges in FL~\cite{yin2021comprehensive,li2021survey,wang2022octopus}. These frameworks can be divided into the categories of homogeneous and heterogeneous FL we introduced above based on their architectures. They differ from the FL frameworks of FedSGD, FedAvg, and FedMD in how the local models or the global model is trained, but the information exchange between the clients and the server is the same as the three FL frameworks, {i.e.,} the clients sharing gradients, model parameters, or predictions on an unlabeled dataset to the server. In this paper, we show the effectiveness of SIAs in the three FL frameworks of FedSGD, FedAvg, and FedMD to demonstrate that an {\em honest-but-curious} server in FL can infer source information of the training records, no matter what kind of updates are uploaded by the clients. But it is worth noting that SIAs are also effective in other FL frameworks because their communication exchange between the clients and the server is the same as the FL frameworks we evaluated in this paper.

\subsection{Membership Inference Attacks}\label{sec::mia}

Membership inference attacks (MIAs) aim to identify whether or not a data record was in the training dataset of a target model. Although an MIA seems like a simple attack, it can directly lead to severe privacy breaches of individuals. For instance, identifying that a patient’s clinical record was used to train a model associated with a disease can reveal that the patient has this disease with a high chance. Although FL has emerged as a popular privacy-aware learning paradigm, recent works~\cite{melis2019exploiting,nasr2019comprehensive,lee2021digestive,zhang2020gan,chen2020beyond,pustozerova2020information,yuan2023interaction} have demonstrated the success of MIAs on FL models. For instance, Melis et al.~\cite{melis2019exploiting} show that a malicious client in FL can infer whether or not a location profile was in the FourSquare dataset that was used to train the global model. In FL, because the training dataset of the FL system consists of all the local training data records, the existing research of MIAs are designed to infer whether or not a data record was used to train the global model, but not to identify whether a data record was used to train a 
local model~\cite{hu2022membership}. 

Currently, there are no attacks in FL to explore which client ({i.e., the source}) owns the training records identified by MIAs, while it is important and practical to explore the source information of the training records. For instance, in a promising FL application where multiple hospitals jointly train a COVID-19 model for COVID-19 diagnosis, an attacker can implement MIAs to infer who has been tested for COVID-19, but further identification of the source hospital where the people are from will make them more prone to discrimination, especially when the hospital is in a high-risk region or country~\cite{devakumar2020racism}. In this paper, we propose SIAs to show the feasibility of breaching the source privacy of the training records in FL. Our proposed SIAs shed new light on how FL reveals sensitive information and the necessity to build more private FL frameworks.

\section{Source Inference Attacks}\label{sec::03}
In this section, we first introduce the threat model. Then, we show how to leverage the Bayesian theorem to analyze how the attacker can perform SIAs based on the prediction loss of the local models.

\subsection{Threat Model}
{
\noindent \textbf{Target FL systems.} As this is the first paper that investigates the source privacy of the training records in the FL system and to avoid the vague definition of SIAs, we consider SIAs on \textit{horizontal FL} (see Section \ref{sec::pre_fl} for detailed introduction of horizontal FL) where there is only one source client for one data record. While in vertical FL, one data record can correspond to multiple source clients, and SIAs under this setting seems more interesting, we leave the investigation of SIAs in vertical FL for our future work.} 

{
We consider the server can be a semi-honest attacker who passively tries to learn the secret data from the communication updates uploaded by the local clients. During the attack process, the server follows the training protocol of the FL system but will passively try to learn the secret by mounting the SIAs. The server can acquire the communication updates uploaded by the local clients. Finally, based on the communication updates, the server can acquire the source information. There is a possibility that the attacker can be one of the clients. However, because local clients in FL can only observe the global model parameters while having little information about the identities of other clients \cite{lyu2022privacy}, it is almost impossible for a local client to achieve the attack goal (detailed in the following paragraph).}

{During the attack process, the local clients follow the training protocol of the FL system. Specifically, the local clients download the global aggregation results calculated by the server and then perform local model updates. After that, the clients upload the necessary information of updates to the server for aggregation.
}

\noindent \textbf{Attack goal.} The attacker of SIAs in FL aims to identify which client owns a training record that participants in the federated training process. The goal of this attack is motivated by FL applications where source information is sensitive. A motivating example is the FL application of multiple hospitals training a COVID-19 diagnosis model (see Section \ref{sec::01} and Section \ref{sec::mia}). Another example is the FL application of multiple users collaboratively training an image classification model. If an attacker can identify which user owns a sensitive image, the attacker can directly obtain the user's privacy based on the sensitive content of that image \cite{melis2019exploiting}.

\noindent \textbf{Attack knowledge.} We consider the attacker of the central server is \textit{honest-but-curious}: The central server will not deviate from the defined FL protocol but will attempt to infer the source information from legitimately received information from the local clients. Specifically, the central server implements SIAs to infer the source information of the clients based on the received gradients, model parameters, or predictions on an unlabeled dataset from the local clients. However, the central server will not actively manipulate these updates from the local clients and thus without affecting the utilities of the FL model. 

Because SIAs are considered as further attacks based on the foundation of MIAs, we follow the similar setting in the literature of MIAs \cite{shokri2017membership,salem2019ml,hu2022membership} that the attacker is given a data record we called \textit{a target record}, and it has been identified as a training record by MIAs. Note that we do not focus on how the attacker obtains the training record from the clients but focus on investigating the potential source privacy leakage of the training record, while one can refer to data reconstruction attacks \cite{zhu2019deep,hitaj2017deep,yin2021comprehensive,fowl2021robbing,boenisch2021curious} to see how an attacker in FL can reconstruct the training data. An SIA is said to succeed if the attacker can correctly identify which local client the target record comes from.

\subsection{Source Inference Attack Method}
In this paper, we focus the horizontal FL on classification tasks. Let $D_\textrm{train}=\{D_1,\cdots,D_K\}$ (assuming there are $K$ clients) be the training dataset of the FL system, where each $D_i$ corresponds to the local training dataset of the client $i$. We assume there are $n$ data records ${\bm{z}_1, \cdots, \bm{z}_n}$ in $D_\textrm{train}$. Each data record is represented as $\bm{z}=(\bm{x},y)$, where $\bm{x}$ is the feature and $y$ is the class label. 

\noindent \textbf{Source status. \;} In horizontal FL, each target record exists in the local dataset of one client. Thus, we can use a $K$-dimensional multinomial vector $\bm{s}$ to represent the source status of each target record. In $\bm{s}$, the $k$-th element equals $1$ representing the target record belongs to the client $k$, while all the remaining elements equal $0$. For example, assuming there are six clients in FL and the target record $\bm{z}$ comes from the second client. Then, the multinomial variable $\bm{s}$ is represented by $\bm{s}=[0,1,0,0,0,0]^{\text{T}}$. 

We assume the target record $\bm{z}_i$ comes from the client $k$ with the probability $\lambda$, i.e., the probability of the $k$-th element in $\bm{s}_i$ equals 1 is $\lambda$, denoted as $\mathbb{P}({{s}}_{ik}=1)=\lambda$. Without loss of generality, we take the case of $\bm{z}_1$ to study the source inference problem, which is defined as follows. 

\begin{definition}[Source Inference]\label{definition::sia}
Given a local model $\bm{\theta}_k$ and a target record $\bm{z}_1$, source inference on $\bm{z}_1$ aims to infer the posterior probability of $\bm{z}_1$ belonging to the client $k$:
\begin{equation}
    \mathcal{S}({\bm{\theta}_k},{\bm{z}_1}): = \mathbb{P}({{s}_{1k}} = 1|{\bm{\theta}_k},{\bm{z}_1}).
\end{equation}
\end{definition}

For source inference by Definition~\ref{definition::sia}, we aim to derive an explicit formula for $\mathcal{S}({\bm{\theta}_k},{\bm{z}_1})$ from the Bayesian perspective, which can provide insights on how to leverage the prediction loss of the local clients for implementing SIAs. We denote $\bm{\tau}=\{\bm{z}_2,\cdots,\bm{z}_n,\bm{s}_2,\cdots,\bm{s}_n\}$ as the set which includes the remaining training records and their source status. The explicit formula of $\mathcal{S}({\bm{\theta}_k},{\bm{z}_1})$ is given by the following theorem.

\begin{theorem}
Given a local model $\bm{\theta}_k$ and a target record $\bm{z}_1$, the source inference is given by:
\begin{equation}
    \mathcal{S}({\bm{\theta}_k},{\bm{z}_1}) = {\mathbb{E}_{\bm{\tau}} }\left[ {{\sigma} \left( {\log (\frac{{\mathbb{P}({\bm{\theta} _k}|{{s}_{1k}} = 1,{\bm{z}_1},\bm{\tau} )}}{{\mathbb{P}({\bm{\theta} _k}|{{s}_{1k}} = 0,{\bm{z}_1},\bm{\tau} )}}) + {\mu _\lambda }} \right)} \right],
\end{equation}
\end{theorem}
\noindent where $\mu_{\lambda}=\log (\frac{\lambda }{{1 - \lambda }})$, and ${\sigma}(\cdot)$ is a sigmoid function defined as ${\sigma} (\bm{x}) = {(1 + {e^{ \bm{- x}}})^{ - 1}}$.

\begin{proof}
Applying the law of total expectation~\cite{montgomery2010applied}, we have:
\begin{equation}
    \begin{aligned}
\mathcal{S}({\bm{\theta}_k},{\bm{z}_1}) &= \mathbb{P}({s_{1k}} = 1|{\bm{\theta}_k},{\bm{z}_1})\\
\quad \quad \quad \;\; &= {\mathbb{E}_{\bm{\tau}} }[\mathbb{P}({s_{1k}} = 1|{\bm{\theta} _k},{\bm{z}_1},\bm{\tau} )].
\end{aligned}
\end{equation}
Applying the Bayes' formula, we have:
\begin{equation}\label{eqn::prob}
    \mathbb{P}({s_{1k}} = 1|{\bm{\theta} _k},{\bm{z}_1},\bm{\tau} ) = \frac{{\mathbb{P}({\bm{\theta} _k}|{s_{1k}} = 1,{\bm{z}_1},\tau )\mathbb{P}({s_{1k}} = 1|{\bm{z}_1},\bm{\tau} )}}{{\mathbb{P}({\bm{\theta} _k}|{\bm{z}_1},\bm{\tau} )}}.
\end{equation}
As the source variables $\bm{s}_1, \cdots, \bm{s}_n$ are independent, event $s_{1k}=1$ is independent from $\bm{z}_1, \bm{\tau}$. Thus, we have:
\begin{equation}\label{eqn::states}
    \mathbb{P}({s_{1k}} = 1|{\bm{z}_1},\bm{\tau} ) = \mathbb{P}({s_{1k}} = 1).
\end{equation}
Let:
\begin{equation}\label{eqn::alpha}
    \phi : = \mathbb{P}({\bm{\theta} _k}|{s_{1k}} = 1,{\bm{z}_1},\bm{\tau} )\mathbb{P}({s_{1k}} = 1).
\end{equation}
\begin{equation}\label{eqn::beta}
    \omega : = \mathbb{P}({\bm{\theta} _k}|{s_{1k}} = 0,{\bm{z}_1},\bm{\tau} )\mathbb{P}({s_{1k}} = 0).
\end{equation}
Plugging eqn. (\ref{eqn::states}), eqn. (\ref{eqn::alpha}), and eqn. (\ref{eqn::beta}) into eqn. (\ref{eqn::prob}), we have:
\begin{equation}
\begin{aligned}
\mathbb{P}({s_{1k}} = 1|{\bm{\theta} _k},{\bm{z}_1},\bm{\tau} ) &= \frac{{\mathbb{P}({\bm{\theta} _k}|{s_{1k}} = 1,{\bm{z}_1},\bm{\tau} )\mathbb{P}({s_{1k}} = 1)}}{{\mathbb{P}({\bm{\theta} _k}|{\bm{z}_1},\bm{\tau} )}}\\ 
\quad \quad \quad \quad \quad \quad \quad \; &= \frac{\phi }{{\phi  + \omega }}\\ 
\quad \quad \quad \quad \quad \quad \quad \; &= \sigma \left( {\log \left(\frac{\phi }{\omega }\right)} \right).
\end{aligned}
\end{equation}
Given that $\mathbb{P}({s}_{ik}=1)=\lambda$, we have:
\begin{equation}
    \log \left( {\frac{\phi }{\omega }} \right) = \log \left( {\frac{{\mathbb{P}({\bm{\theta} _k}|{s_{1k}} = 1,{\bm{z}_1},\bm{\tau} )}}{{\mathbb{P}({\bm{\theta} _k}|{s_{1k}} = 0,{\bm{z}_1},
    \bm{\tau} )}}} \right) + \log \left( {\frac{\lambda }{{1 - \lambda }}} \right).
\end{equation}
Let $\mu_{\lambda}=\log (\frac{\lambda }{{1 - \lambda }})$, we have:
\begin{equation}
\begin{aligned}
    \mathcal{S}({\bm{\theta}_k},{\bm{z}_1}) &= {\mathbb{E}_{\bm{\tau}} }[\mathbb{P}({s_{1k}} = 1|{\bm{\theta} _k},{\bm{z}_1},\bm{\tau} )] \\
    &= {\mathbb{E}_{\bm{\tau}} }  \left[ { \sigma \left( {\log \left(\frac{\phi }{\omega }\right)} \right) } \right]\\
    &= {\mathbb{E}_{\bm{\tau}} }\left[ {{\sigma} \left( {\log (\frac{{\mathbb{P}({\bm{\theta} _k}|{{s}_{1k}} = 1,{\bm{z}_1},\bm{\tau} )}}{{\mathbb{P}({\bm{\theta} _k}|{{s}_{1k}} = 0,{\bm{z}_1},\bm{\tau} )}} + {\mu _\lambda }} \right)} \right],
\end{aligned}
\end{equation}
which concludes the proof.
\end{proof}
We observe that \textit{Theorem 1} does not have the loss $\ell(\cdot)$ form and only relies on the posterior parameter $\bm{\theta}_k$ in expectation given $\{\bm{z}_1,\cdots,\bm{z}_n,\bm{s}_1,\cdots,\bm{s}_n \}$ is a random variable. To make $\mathcal{S}({\bm{\theta}_k},{\bm{z}_1})$ more explicit with the loss term, we assume an ML algorithm produced parameters $\bm{\theta}$ follows a posterior distribution. Specifically, following the assumption in the previous work~\cite{sablayrolles2019white}, we assume the posterior distribution of an ML model $\bm{\theta}$ as follows:
\begin{equation}
    p(\bm{\theta} |{\bm{z}_1}, \cdots ,{\bm{z}_n}) \propto {e^{ - \frac{1}{\gamma }\sum\nolimits_{i = 1}^n {\ell (\bm{\theta} ,{\bm{z}_i})} }},
\end{equation}

\noindent where $\gamma$ is a temperature parameter which controls the stochasticity of $\bm{\theta}$. Following this assumption, given $\{\bm{z}_1,\cdots,\bm{z}_n,\bm{s}_1,\cdots,\bm{s}_n \}$, the posterior distribution of $\bm{\theta}_k$ follows:
\begin{equation}
    p({\bm{\theta} _k}|{\bm{z}_1}, \cdots ,{\bm{z}_n},{\bm{s}_1}, \cdots ,{\bm{s}_n}) \propto {e^{ - \frac{1}{\gamma }\sum\nolimits_{i = 1}^n {{s_{ik}}\ell (\bm{\theta}_k ,{\bm{z}_i})} }}.
\end{equation}
We further define the posterior distribution of $\bm{\theta}_k$ given training records $\bm{z}_2,\cdots,\bm{z}_n$ and their source status $\bm{s}_2,\cdots,\bm{s}_n$ :
\begin{equation}
    {p_{\bm{\tau}} }({\bm{\theta} _k}): = \frac{{{e^{ - \frac{1}{\gamma }\sum\nolimits_{i = 2}^n {{s_{ik}}\ell (\bm{\theta}_k ,{\bm{z}_i})} }}}}{{\int_{\bm{t}} {{e^{ - \frac{1}{\gamma }\sum\nolimits_{i = 2}^n {{s_{ik}}\ell (\bm{t} ,{\bm{z}_i})} }}} d\bm{t}}},
\end{equation}
where the denominator is a constant value. The following theorem explicitly demonstrates how to conduct the source inference with the loss term.
\begin{theorem}\label{theo::opt}
Given local resulting model $\bm{\theta}_k$ and a target record $\bm{z}_1$, the source inference attack is given by:
\begin{equation}\label{eqn::explicit-sia}
    \mathcal{S}({\bm{\theta} _k},{\bm{z}_1}) = {\mathbb{E}_{\bm{\tau}} }\left[ {\sigma \left( {g({\bm{z}_1},\bm{\theta} ,{p_{\bm{\tau}} }) + {\mu _\lambda }} \right)} \right],
\end{equation}

\noindent where
\begin{align}
{\ell _{{p_{\bm{\tau}} }}}({\bm{z}_1}): &=  - \gamma \log \left( {\int_{\bm{t}} {{e^{ - \frac{1}{\gamma }\ell (\bm{t},{\bm{z}_1})}}} {p_{\bm{\tau}} }(\bm{t})d\bm{t}} \right) \label{loss_other},\\
\ell (\bm{\theta}_k,{\bm{z}_1}): &=  - \gamma \log \left( {{e^{ - \frac{1}{\gamma }\ell ({\bm{\theta}_k},{\bm{z}_1})}}} \right) \label{loss_k}, \\
g({\bm{z}_1},\bm{\theta} ,{p_{\bm{\tau}} }): &= \frac{1}{\gamma }({\ell _{{p_{\bm{\tau}} }}}({\bm{z}_1}) - \ell ({\bm{\theta} _k},{\bm{z}_1})). \label{gap}
\end{align}
\end{theorem}

\begin{proof}
For $\phi$ and $\omega$ defined in eqn.~(\ref{eqn::alpha}) and eqn.~(\ref{eqn::beta}), we have:
\begin{equation}
\begin{aligned}
\phi  &= \lambda \frac{{{e^{ - \frac{1}{\gamma }\ell ({\bm{\theta}_k},{\bm{z}_1})}}{e^{ - \frac{1}{\gamma }\sum\nolimits_{i = 2}^n {{s_{ik}}\ell ({\bm{\theta}_k},{\bm{z}_i})} }}}}{{\int_{\bm{t}} {{e^{ - \frac{1}{\gamma }\ell (t,{\bm{z}_1})}}{e^{ - \frac{1}{\gamma }\sum\nolimits_{i = 2}^n {{s_{ik}}\ell (\bm{t},{\bm{z}_i})} }}d\bm{t}} }}\\
 &= \lambda \frac{{{e^{ - \frac{1}{\gamma }\ell ({\bm{\theta} _k},{\bm{z}_1})}}{p_{\bm{\tau}} }({\bm{\theta} _k})}}{{\int_{\bm{t}} {{e^{ - \frac{1}{\gamma }\ell (\bm{t},{\bm{z}_1})}}{p_{\bm{\tau}} }(\bm{t})d\bm{t}} }},
\end{aligned}
\end{equation}
\begin{equation}
\begin{aligned}
\omega  &= (1 - \lambda )\frac{{{e^{ - \frac{1}{\gamma }\sum\nolimits_{i = 2}^n {{s_{ik}}\ell ({\bm{\theta}_k},{\bm{z}_i})} }}}}{{\int_{\bm{t}} {{e^{ - \frac{1}{\gamma }\sum\nolimits_{i = 2}^n {{s_{ik}}\ell (\bm{t},{\bm{z}_i})} }}d\bm{t}} }}\\
 &= (1 - \lambda ){p_{\bm{\tau}} }(\bm{\theta}_k).
\end{aligned}
\end{equation}
Thus, we have:
\begin{equation}
\begin{aligned}
    \log \left( {\frac{\phi }{\omega }} \right) &=  - \log \left( {\int_{\bm{t}} {{e^{ - \frac{1}{\gamma }\ell (\bm{t},{\bm{z}_1})}}{p_{\bm{\tau}} }(\bm{t})d\bm{t}} } \right) + \log \left( {{e^{ - \frac{1}{\gamma }\ell ({\bm{\theta}_k},{\bm{z}_1})}}} \right)\\
 &\mathrel{\phantom{=}}+ \log \left( {\frac{\lambda }{{1 - \lambda }}} \right)\\
 &= \frac{1}{\gamma }({\ell _{{p_{\bm{\tau}} }}}({\bm{z}_1}) - \ell \left({\bm{\theta} _k},{\bm{z}_1})\right) + {\mu _\lambda }.
\end{aligned}
\end{equation}
\begin{equation}
\begin{aligned}
    \mathcal{S}({\bm{\theta}_k},{\bm{z}_1}) &= {\mathbb{E}_{\bm{\tau}} }  \left[ { \sigma \left( {\log \left(\frac{\phi }{\omega }\right)} \right) } \right]\\
    &= {\mathbb{E}_{\bm{\tau}} }\left[ { {\sigma} \left(  \frac{1}{\gamma }({\ell _{{p_{\bm{\tau}} }}}({\bm{z}_1}) - \ell \left({\bm{\theta} _k},{\bm{z}_1})\right) + {\mu _\lambda } \right) } \right] \\
    &= {\mathbb{E}_{\bm{\tau}} }\left[ {\sigma \left( {g({\bm{z}_1},\bm{\theta} ,{p_{\bm{\tau}} }) + {\mu _\lambda }} \right)} \right],
\end{aligned}
\end{equation}
which concludes the proof.
\end{proof}

\begin{figure*}
    \centering
    \includegraphics[height=2.8in,width=0.80\linewidth]{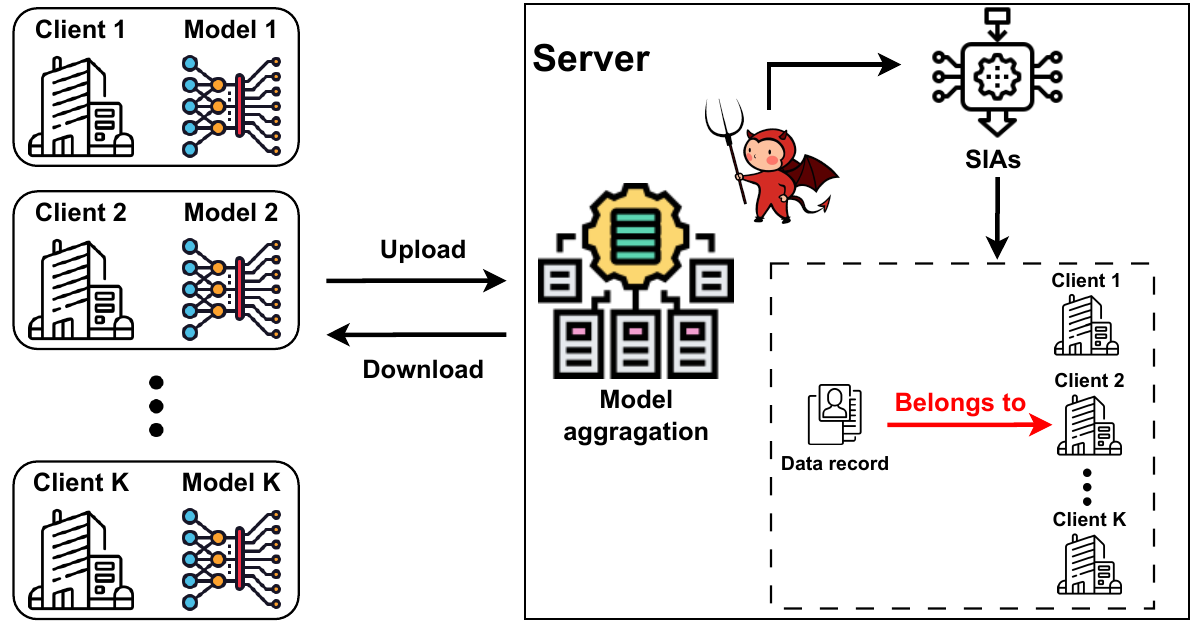}
    \caption{{An overview of SIAs in FL. In each communication round, each client transmits the necessary information of updates to the central server for aggregation. The central server faithfully follows the defined FL protocol while inferring the source clients of the target data records from legitimately received information from the local clients.}}
    \label{fig::sia}
\end{figure*}

\subsection{Analysis of the Source Inference Attack}\label{sec::ana_sias}
Theorem~\ref{theo::opt} implies that an attacker can infer the posterior probability of a target record belonging to a particular client by leveraging its prediction loss of the client's local model. Now we analyze the meaning of each term in eqn.~(\ref{eqn::explicit-sia}) and discuss how they affect the posterior probability. We conclude from Theorem~\ref{theo::opt} and derive an attack method that leverages the prediction loss of the local models to infer the source client of a target record.

\noindent \textbf{Meaning of $g({\bm{z}_1},\bm{\theta} ,{p_{\bm{\tau}} })$. \;} The term $g({\bm{z}_1},\bm{\theta} ,{p_{\bm{\tau}} })$ is the gap between $\ell _{{p_{\bm{\tau}} }}({\bm{z}_1})$ and $\ell \left({\bm{\theta} _k},{\bm{z}_1}\right)$. Because $\bm{\tau}$ is a training set that does not contain any information about $\bm{z}_1$, $p_{\bm{\tau}}$ corresponds to a posterior distribution of an ML model's parameters that trained without using $\bm{z}_1$. Note that $\ell(\cdot)$ is a loss function that measures the performance of the model on a data record. $\ell \left({\bm{\theta} _k},{\bm{z}_1}\right)$ is the local model $\bm{\theta}_k$'s evaluation of the loss on the target record $\bm{z}_1$. Comparing eqn.~(\ref{loss_other}) and eqn.~(\ref{loss_k}), we can find that $\ell _{{p_{\bm{\tau}} }}({\bm{z}_1})$ is the expectation of the loss $\ell (\cdot,\bm{z}_1)$ over the typical models that have not seen $\bm{z}_1$. Thus, we can interpret $g({\bm{z}_1},\bm{\theta} ,{p_{\bm{\tau}} })$ as the difference between $\bm{\theta}_k$'s loss on $\bm{z}_1$ and other models' (trained without $\bm{z}_1$) average loss on $\bm{z}_1$. 

In other words, $g({\bm{z}_1},\bm{\theta} ,{p_{\bm{\tau}} })$ is the differences between the prediction loss of the local model of client $k$ and the average prediction loss of other clients' local models. If $\ell _{{p_{\bm{\tau}} }}({\bm{z}_1}) \approx \ell \left({\bm{\theta} _k},{\bm{z}_1}\right)$, which means the client $k$ behaves almost the same as other clients on $\bm{z}_1$, then $g({\bm{z}_1},\bm{\theta} ,{p_{\bm{\tau}} }) \approx 0$. Since $\sigma(\mu_{\lambda})=\lambda$, the posterior probability $\mathcal{S}({\bm{\theta} _k},{\bm{z}_1})$ is equal to $\lambda$. Thus, we have no source information gain on $\bm{z}_1$ beyond prior knowledge. In FL, the prior knowledge $\mathbb{P}({s}_{ik}=1)=\lambda=\frac{1}{K}$. In this case, the source inference equals \textit{random guess}. However, if $\ell _{{p_{\bm{\tau}} }}({\bm{z}_1}) > \ell \left({\bm{\theta} _k},{\bm{z}_1}\right)$, that is, the client $k$ performs better than other clients on $\bm{z}_1$, $g({\bm{z}_1},\bm{\theta} ,{p_{\bm{\tau}} })$ becomes positive. When $g({\bm{z}_1},\bm{\theta} ,{p_{\bm{\tau}} })>0$, $\mathbb{P}({{s}_{1k}} = 1|{\bm{\theta}_k},{\bm{z}_1})>\lambda$ and thus we gain non-trivial source information on $\bm{z}_1$. Moreover, since $\sigma(\cdot)$ is non decreasing, smaller $ \ell \left({\bm{\theta} _k},{\bm{z}_1})\right)$ indicates higher probability that $\bm{z}_1$ belongs to the client $k$. 

\noindent \textbf{Conclusion from Theorem~\ref{theo::opt}. \;} We conclude that the smaller the loss of client $k$'s local model on a target record $\bm{z}_1$, the higher posterior probability that $\bm{z}_1$ belongs to the client $k$. This motivates us to design the source inference attack that the client whose local model has the smallest loss on a target record should own this record. Moreover, if the client's local model's behavior on its local training data is different from that of other clients, our attack will always achieve better performance than randomly guessing. We give more empirical evidence in Section~\ref{sec::04}.

\subsection{Source Inference Attacks in Different FL Frameworks}
In this paper, we investigate SIAs in three FL frameworks under horizontal FL. Specifically, we investigate SIAs in FedSGD \cite{mcmahan2017communication}, FedAvg \cite{mcmahan2017communication}, and FedMD \cite{li2019fedmd} where local clients upload gradients, model parameters, or predictions on an unlabeled dataset to the server. {The success of SIAs in the three FL frameworks sheds light on how the communications between clients and the semi-honest central server in existing FL frameworks enable the server to mount SIAs to steal source information.} Fig.~\ref{fig::sia} shows an overview of how the honest-but-curious server implements SIAs in FL.

Intuitively, the server in FedAvg can directly conduct an SIA in each communication round because it receives the parameters of local models from the clients. Thus, the server can directly use the local models to calculate the prediction loss of a target record for implementing source inference. However, in FedSGD and FedMD, the server cannot directly implement SIAs because it cannot directly leverage the updates from the clients to calculate the prediction loss of the local models on the target records. In these two frameworks, we introduce two strategies to enable the server to implement the proposed SIAs.

In FedSGD, each client $k$ uploads the average gradient ${\bm{g}_k} = \nabla {\ell }({\bm{\theta} _{t-1}},D_k)$ calculated on its local data $D_k$ at the current global model $\bm{\theta}_{t-1}$. Thus, the server can leverage the gradient from each client to update the global model $\bm{\theta} _{t}^k \leftarrow {\bm{\theta} _{t-1}} - \eta {\bm{g}_k}$ separately, where $\eta$ is a fixed learning rate in the FL framework. Note that $\bm{\theta} _{t}^k$ essentially is the updated local model of the client $k$ in the $t$-th communication round. Thus, the server can use $\bm{\theta} _{t}^k$ to calculate the prediction loss of each local model in each communication round to conduct SIAs. 

{In FedMD, the server can leverage knowledge distillation \cite{ba2014deep,hinton2015distilling} to achieve SIAs. Knowledge distillation aims to transfer knowledge of larger models to smaller models so that the smaller models are as accurate as larger models. The larger models are referred to as teacher models and the smaller models are referred to as students models. During knowledge distillation, the student model is trained to match the logits of the teacher model for learning its knowledge. Knowledge distillation enables the smaller student model to have similar performance to their teacher models \cite{crowley2018moonshine}. In FedMD, the server cannot directly use the updates from the clients to calculate the prediction loss of the target records because the updates are predictions of a public dataset. However, because the clients' predictions on the public dataset represent the knowledge of the local models, the server can leverage knowledge distillation to mount SIAs. Specifically, for each of the local clients, the server considers it as a teacher and leverages its predictions to train a student model to mimic the local model. Because the student models are expected to behave similarly to the local models, the server can use the prediction loss of the student models on the target records as the estimate of the prediction loss of the local models. Thus, based on the estimated prediction loss, the semi-honest server can mount SIAs in FedMD. Note that although the student models mimic the local models, there is a bias between the estimated prediction loss and the actual loss of the local model. However, if the estimated prediction loss of the client owning the target instance is distinguishable from the estimated losses of other clients, the SIAs can still succeed, which we will show in the experiments.}

Based on the analysis above, we propose three FL frameworks \fedsgd, \fedavg, and \fedmd \;that allow an honest-but-curious server to conduct SIAs without deviating from the normal FedSGD, FedAvg, or FedMD protocols. 
Algorithm \ref{alg::fedsgd}, Algorithm \ref{alg::fedavg}, and Algorithm \ref{alg::fedmd} describe \fedsgd, \fedavg, and \fedmd, respectively. Each algorithm consists of two steps, {i.e.,} \textit{Server executes} and \textit{ClientUpdate}. In each algorithm, we assume there are $K$ clients and we take a target record of $\bm{z}$ as an example to show how to implement SIAs.

\noindent \textbf{\fedsgd: \;} i) \textbf{Server executes:} As depicted in Algorithm 1, the honest-but-curious server faithfully follows FedSGD protocol (Lines 2-5, 8, 10, 11, 13, and 14) while implementing SIAs (Lines 6, 7, and 9). First, the server initializes the weights of the global model randomly following the step in Line 2. Next, in each communication round, the server will receive the gradient calculated on the local private dataset from each client (Lines 3-5). Then, the server can use the gradient from each client to update the global model separately to obtain the local models (Line 6). The server calculates each local model's prediction loss on $\bm{z}$ and obtains the source $i$ of $\bm{z}$ by finding which client has the smallest loss on $\bm{z}$ (Line 7 and 8). Last, the server averages the gradients to update a new global model (Line 10), which is then distributed for the next round of updating. ii) \textbf{ClientUpdate:} Each client has its own private dataset. In each communication round, the client calculates the gradient based on the private dataset and the current global model (Line 13). Then, the gradients are sent back to the server (Line 14).

\begin{algorithm}[t!]
\caption{\fedsgd\: The $K$ clients are indexed by $k$; $\eta$ represents the learning rate; $\bm{z}$ represents a target record.}
\label{alg::fedsgd}
\begin{algorithmic}[1]
\State \textbf{\hskip 9em Server executes}

\State initialize $\theta_{0}$ \Comment{\fade{initialize weights of the global model}}
\For{each round $t = 1$ to $T$}
\For{each client $k$}
\State $g^{k}_{t}$ $\leftarrow$ \textbf{ClientUpdate}($\theta_{t-1}$) \Comment{\fade{local gradient of the client k at round t}}
\State ${\theta} _{t}^k \leftarrow {{\theta} _{t-1}} - \eta {{g}_k}$ \Comment{\fade{use the gradient of the client to obtain the local model}}
\State Compute $\ell({\theta^{k}_{t},{\bm{z}})}$ \Comment{\fade{calculate the local prediction loss on ${\bm{z}}$}}
\EndFor

\State $i \leftarrow argmin (\ell({\theta}_1,{\bm{z}}) ,\cdots,\ell({\theta}_K,{\bm{z}}))$ \Comment{\fade{identify the source client}}

\State ${\theta _t} \leftarrow {\theta _{t - 1}} - \eta \sum\nolimits_k {g_t^k}$ \Comment{\fade{update the global model}}
\EndFor

\State \textbf{\hskip 9em ClientUpdate}($\theta$) \Comment{\fade{run on client k}}
\State Compute ${{g}_k} = \nabla {\ell }({{\theta}},D_k)$

\State \Return ${{g}_k}$ \Comment{\fade{return the gradient to the central server}}

\end{algorithmic}
\end{algorithm}

\begin{algorithm}[t!]
\caption{\fedavg\: The $K$ clients are indexed by $k$; $B$ represents the local mini-batch size; $E$ represents the number of local epochs; $\eta$ represents the learning rate; $\bm{z}$ represents a target record.}
\label{alg::fedavg}
\begin{algorithmic}[1]
\State \textbf{\hskip 9em Server executes}

\State initialize $\theta_{0}$ \Comment{\fade{initialize weights of the global model}}
\For{each round $t = 1$ to $T$}
\For{each client $k$}
\State $\theta^{k}_{t}$ $\leftarrow$ \textbf{ClientUpdate}($\theta_{t-1}$) \Comment{\fade{local model weight of client k at round t}}
\State Compute $\ell({\theta^{k}_{t},{\bm{z}})}$ \Comment{\fade{calculate the local prediction loss on ${\bm{z}}$}}
\EndFor

\State $i \leftarrow argmin (\ell({\theta}_1,{\bm{z}}) ,\cdots,\ell({\theta}_K,{\bm{z}}))$ \Comment{\fade{identify the source client}}

\State ${\theta _t} \leftarrow \sum\nolimits_k {\frac{{{n^{(k)}}}}{n}} \theta _t^k$ \Comment{\fade{update the global model}}
\EndFor

\State \textbf{\hskip 9em ClientUpdate}($\theta$) \Comment{\fade{run on client k}}

\State  $\mathcal{B}$ $\leftarrow$ (split $D_{k}$ into multiple batches of size $B$) 
\For{each local epoch $i$ from $1$ to $E$}
\For{batch $b \in \mathcal{B}$}
\State $\theta  \leftarrow \theta  - \eta \nabla \ell (b;\theta )$ \Comment{\fade{mini-batch gradient descent}}
\EndFor
\EndFor

\State \Return $\theta$ \Comment{\fade{return local model to the central server}}

\end{algorithmic}
\end{algorithm}

\begin{algorithm}[t!]
\caption{\fedmd\: The $K$ clients are indexed by $k$; $\bm{z}$ represents a target record; $D_0$ is a public dataset; $D_k$ is a private dataset of the client $k$; $\theta_k$ represents the local model of the client k; $\theta_s$ is a student model.}
\label{alg::fedmd}
\begin{algorithmic}[1]
\State \textbf{\hskip 9em Initialization phase}
\State the clients initialize local models $\theta_1,\cdots,\theta_K$ and each client k updates her model weight $\theta_k$ on her private dataset $D_k$
\For{each client $k$}
\State $Y_0^k = \textsc{\textsf{\small{Predict}}}(\theta_k,D_0)$
\State send $Y_0^k$ to the server
\EndFor
\State server calculates ${{\tilde Y}_0} = \frac{1}{K}\sum\nolimits_k {Y_0^k}$ 

\vspace{1em}

\State \textbf{\hskip 9em Server executes}

\For{each round $t = 1$ to $T$}
\For{each client $k$}
\State $Y^{k}_{t}$ $\leftarrow$ \textbf{ClientUpdate}($\theta_k,{{\tilde Y}_{t-1}},D_0$) \Comment{\fade{predictions of client k on $D_0$ at round t}}
\State $\theta_s^k \leftarrow$ \textsc{\textsf{\small{Train}}}$(\theta_s,Y^{k}_{t},D_0)$ \Comment{\fade{server trains a student model to mimic the local model}}
\State Compute $\ell({\theta^{k}_{s},{\bm{z}})}$ \Comment{\fade{calculate local loss on ${\bm{z}}$ using the student model}}
\EndFor
\EndFor
\State $i \leftarrow argmin (\ell({\theta}_s^1,{\bm{z}}) ,\cdots,\ell({\theta}_s^K,{\bm{z}}))$ \Comment{\fade{identify the source client}}
\State ${{\tilde Y}_t} = \frac{1}{K}\sum\nolimits_k {Y_t^k}$ \Comment{\fade{prediction aggregation at the server}}

\vspace{0.5em}

\State \textbf{\hskip 9em ClientUpdate}($\theta,{{\tilde Y}_{t-1}},D_0$) \Comment{\fade{run on client k}}
\For{each local epoch $i$ from $1$ to $E_1$}
\State Digest: $\theta \leftarrow$ \textsc{\textsf{\small{Train}}}$(\theta,{\tilde Y}_{t-1},D_0)$ \Comment{\fade{train the model on the public dataset}}
\EndFor
\For{each local epoch $i$ from $1$ to $E_2$}
\State Revisit: $\theta \leftarrow$ \textsc{\textsf{\small{Train}}}$(\theta,D_k)$ \Comment{\fade{train the model on the private dataset}}
\EndFor
\State $Y = \textsc{\textsf{\small{Predict}}}(\theta,D_0)$ \Comment{\fade{predictions of class scores on $D_0$}}

\State \Return $Y$ \Comment{\fade{return predictions to the central server}}

\end{algorithmic}
\end{algorithm}

\noindent \textbf{\fedavg: \;} i) \textbf{Server executes:} As depicted in Algorithm \ref{alg::fedavg}, the honest-but-curious server faithfully follows FedAvg protocol (Lines 2-5, 7, 9, 10, 12-18) while implementing SIAs (Lines 6 and 8). First, the server initializes the weights of the global model randomly (Line 2). Next, in each communication round, the server will receive the updated models from clients (Lines 3, 4, and 5) and calculates each model's prediction loss on $\bm{z}$ (Line 6). The server then obtains the source $i$ of $\bm{z}$ by finding which client has the smallest loss on $\bm{z}$ (Line 8). Last, the server averages the uploaded local models to allocate a new global model as usual (Line 9). ii) \textbf{ClientUpdate:} In each communication round, the local clients will update the global model distributed from the server. Each client performs local updates by using mini-batch gradient descent to update the local models' weights (Lines 12-17). To minimize communication with the server, clients update local models for several epochs. Then, the updated models are sent back to the server (Line 18).

\noindent \textbf{\fedmd: \;} i) \textbf{Server executes:} As depicted in Algorithm \ref{alg::fedmd}, the honest-but-curious server faithfully follows FedMD protocol (Lines 2-7, 9-11, 14, 15, 17, and 19-26) while implementing SIAs (Lines 12, 13, and 16). In each communication round, the server receives the predictions of the public dataset from each client (Lines 9-11). Then, for each client, the server trains a student model on the public dataset to imitate the local model (Line 12). The server calculates each student model's prediction loss on $\bm{z}$ and obtains the source $i$ by finding which student model has the smallest prediction loss on $\bm{z}$ (Lines 13 and 16). Last, the server aggregates the predictions from each client for the next round of updating (Line 17). ii) \textbf{ClientUpdate:} First, each client trains the local model on the soft-labeled public dataset to approach the consensus on the public dataset (Lines 19-21 for Digest). Then, each client trains the local model on the private dataset (Lines 22-25 for Revisit). Last, each client computes the predictions on the public dataset and sends the predictions to the server (Line 26).

\noindent \textbf{Complexity analysis of SIAs. \;} The proposed SIAs leverage the prediction loss of local models to infer the source client of a target record. The computational complexity of SIAs mainly depends on two factors. One is the evaluation of the prediction loss of local models, and the other is the identification of the client having the smallest prediction loss. In \fedavg, because the server directly leverages the uploaded local models to evaluate the prediction loss, the computation cost is determined by the model size. The time complexity of SIAs in \fedavg\; with respect to the number of clients $K$ is $O(K)$. As the server received the gradients from each client to update the global model in \fedsgd, the local models can be obtained by the server for SIAs with no extra computation cost. Thus, the computational complexity of \fedsgd\;is the same as \fedavg. In \fedmd, SIAs are more complicated than that in \fedsgd\;and \fedavg\;because the server requires a student model to mimic the behavior of local models. The computation cost of training the student models is mainly determined by the student model size and the size of the public dataset. Then, the server can implement SIAs based on the student models, with a time complexity of $O(K)$ with respect to the number of clients. In our experiments using PyTorch with a single GPU NVIDIA Tesla P40, the execution time of SIAs in \fedsgd\;and \fedavg\;is within seconds and in \fedmd\;is within one minute.

\section{Experiments}\label{sec::04}
In this section, we empirically evaluate \fedsgd, \fedavg, and \fedmd. We first introduce datasets and model architectures used in the experiments. Then, we demonstrate the effectiveness of SIAs in the three FL frameworks and conduct a detailed ablation study to identify how different factors in FL influences the performance of SIAs. In the end, we discuss why our SIAs work.
\subsection{Datasets and Model Architectures}

\noindent \textbf{Datasets. \;} The datasets used in our experiments are reported in Table~\ref{tab::datasets}. We create an IID \textit{Synthetic} dataset to allow us to precisely manipulate data heterogeneity. We generate \textit{Synthetic} as described in previous works~\cite{li2020fedprox,li2019convergence}. The remaining datasets are widely used datasets for simulating and evaluating the privacy leakage on machine learning models~\cite{shokri2017membership,jayaraman2019evaluating,ganju2018property,wang2019beyond}. For MNIST and CIFAR10, the training dataset and testing datasets have been divided when downloading them. For the remaining datasets, we use the \textit{train\_test\_split} function from the \textit{sklearn}~\footnote{\url{https://scikit-learn.org/stable/}} toolkit to randomly select $80$\% samples as the training records (before partitioning client data), and the remaining $20$\% records are used as the testing records. We use the four datasets to evaluate \fedsgd \;and \fedavg. Because \fedmd \;requires a public dataset to share the knowledge of the clients, we evaluate it on paired datasets. Specifically, we select paired MNIST/FEMNIST and CIFAR-10/CIFAR-100. For MNIST/FEMNIST pairs, we select MNIST as the public data and a subset of the Federated Extended MNIST (FEMNIST)~\cite{caldas2018leaf} as the private data. For CIFAR-10/CIFAR-100 pairs, CIFAR-10 is selected as the public dataset, and the private dataset is a subset of CIFAR-100. 

\begin{table}[t!]
\caption{A summary of datasets used in the experiments.}
\label{tab::datasets}
\centering
\begin{tabular}{lllc}
\toprule
\bfseries Dataset & \bfseries \#Records & \bfseries \#Classes & \bfseries Dimension of records\\
\midrule
Synthetic & 100k & 10 & 60 \\

 MNIST~\tablefootnote{http://yann.lecun.com/exdb/mnist/} & 70k & 10 & 1x28x28 \\

 CIFAR-10~\tablefootnote{https://www.cs.toronto.edu/~kriz/cifar.html} & 60k & 10 & 3x32x32\\

 FEMNIST~\tablefootnote{https://github.com/TalwalkarLab/leaf} & 80k & 62 & 1x28x28 \\

 CIFAR-100~\tablefootnote{https://www.cs.toronto.edu/~kriz/cifar.html} & 60k & 100 & 3x32x32\\

 Purchase~\tablefootnote{https://www.kaggle.com/c/acquire-valued-shoppers-challenge/data} & 197.3k & 100 & 600\\
\bottomrule
\end{tabular}
\end{table}

\begin{table}[!t]
    \caption{The DNN architecture is used for classification tasks. (a) The CNN architecture is used for MNIST, and CIFAR-10. (b) The FC architecture is used for Synthetic and Purchase.}
    \label{tab::model}
    \begin{subtable}{.5\linewidth}
      \centering
        \caption{CNN architecture}
        \label{cnn}
        \begin{tabular}{l|l}
        \toprule
        Layer type & Size \\
        \midrule
        \multicolumn{2}{c}{Input} \\
        Convolution + ReLU & 5$\times$5$\times$32 \\
        
        Max Pooling & 2$\times$2 \\
       
        Convolution + ReLU & 5$\times$5$\times$64 \\
        
        Max Pooling & 2$\times$2 \\
       
        FC + ReLU & 512 \\
        
        FC + ReLU & 128 \\
        
        Activation & Softmax \\
        \multicolumn{2}{c}{Output} \\
        \bottomrule
        \end{tabular}
    \end{subtable}%
    \begin{subtable}{.5\linewidth}
      \centering
        \caption{MLP architecture}
        \label{fc}
        \begin{tabular}{l|l}
        \toprule
        Layer type & Size \\
        \midrule
        \multicolumn{2}{c}{Input}\\
        
        FC + ReLU & 200 \\
        
        Activation & Softmax \\
        
        \multicolumn{2}{c}{Output} \\
        \bottomrule
        \end{tabular}
    \end{subtable} 
\end{table}

\begin{figure*}[t!]
    \centering
  \subfloat[$\alpha=100$\label{1a}]{%
       \includegraphics[width=0.33\linewidth]{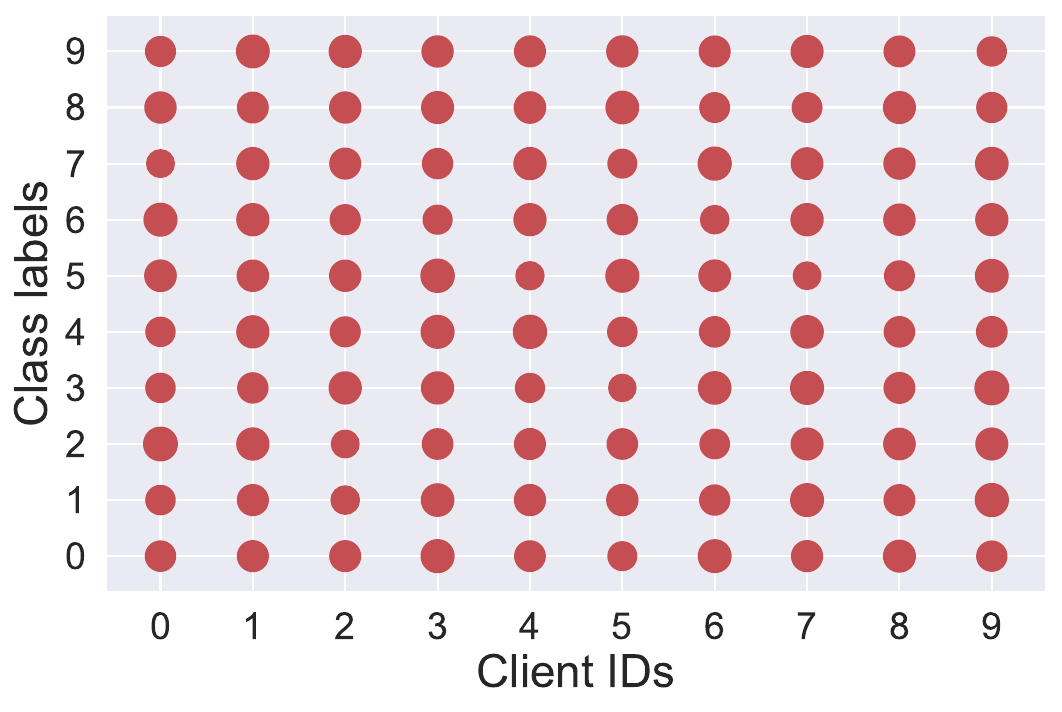}}
  \subfloat[$\alpha=1$\label{1b}]{%
        \includegraphics[width=0.33\linewidth]{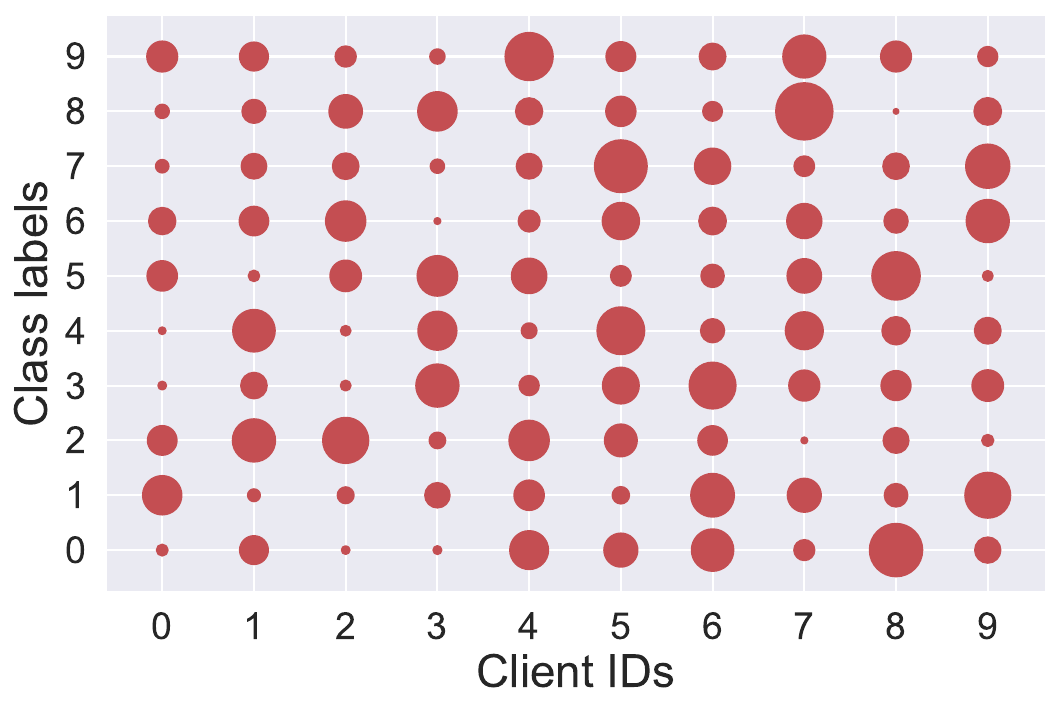}}
  \subfloat[$\alpha=0.1$\label{1c}]{%
        \includegraphics[width=0.33\linewidth]{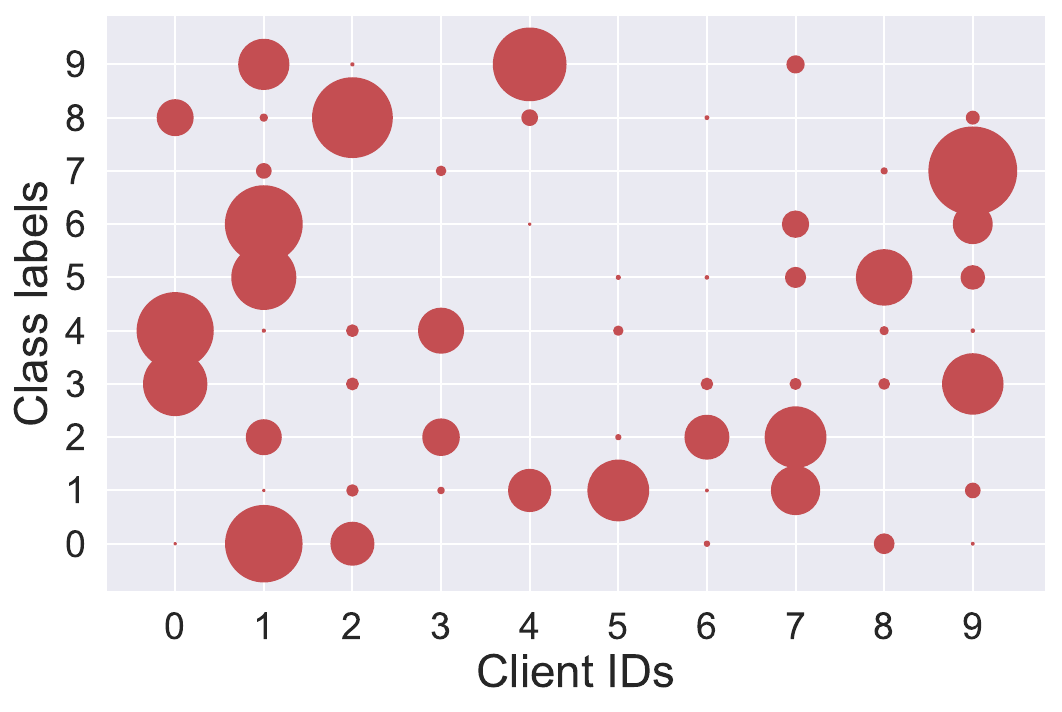}}
  \caption{Illustration of the number of samples per class allocated to each client at different Dirichlet distribution alpha values, for CIFAR-10 with 10 clients. The x-axis represents the IDs of the clients. The y-axis represents the class labels of CIFAR-10. The dot size reflects the number of samples allocated to the clients. As we can see, the data distribution across the clients becomes more and more non-IID as $\alpha$ decreases.}
  \label{dirichlet} 
\end{figure*}

\noindent \textbf{Models. \;} We use deep neural networks (DNNs) as the global models in FL for the classification tasks. Specifically, we use convolutional neural networks (CNN) for the image datasets of MNIST and CIFAR-10. For binary datasets of Synthetic and Purchase, we use fully-connected (FC) neural networks. The architecture of the CNN and FC classifiers used in \fedsgd \;and \fedavg \;are described in Table~\ref{tab::model}. Because \fedmd \;is designed for FL with heterogeneous architectures, we adopt the setting of the different CNN architectures of the clients in \cite{li2019fedmd}, which proposed the FedMD framework. For the detailed description of the CNN architectures of the clients in \fedmd, readers can refer to \cite{li2019fedmd}. The student model in \fedmd \;used for imitating the local models is listed in Table~\ref{cnn}. Note that the DNN architectures used in this paper do not necessarily achieve the best performance for the considered datasets in FL, because our goal is not to attack the best DNN architecture. In this paper, we aim to show that FL is vulnerable to SIAs. 

\begin{figure*}[t!]
    \centering
  \subfloat[\fedsgd \label{3a}]{%
       \includegraphics[width=0.33\linewidth]{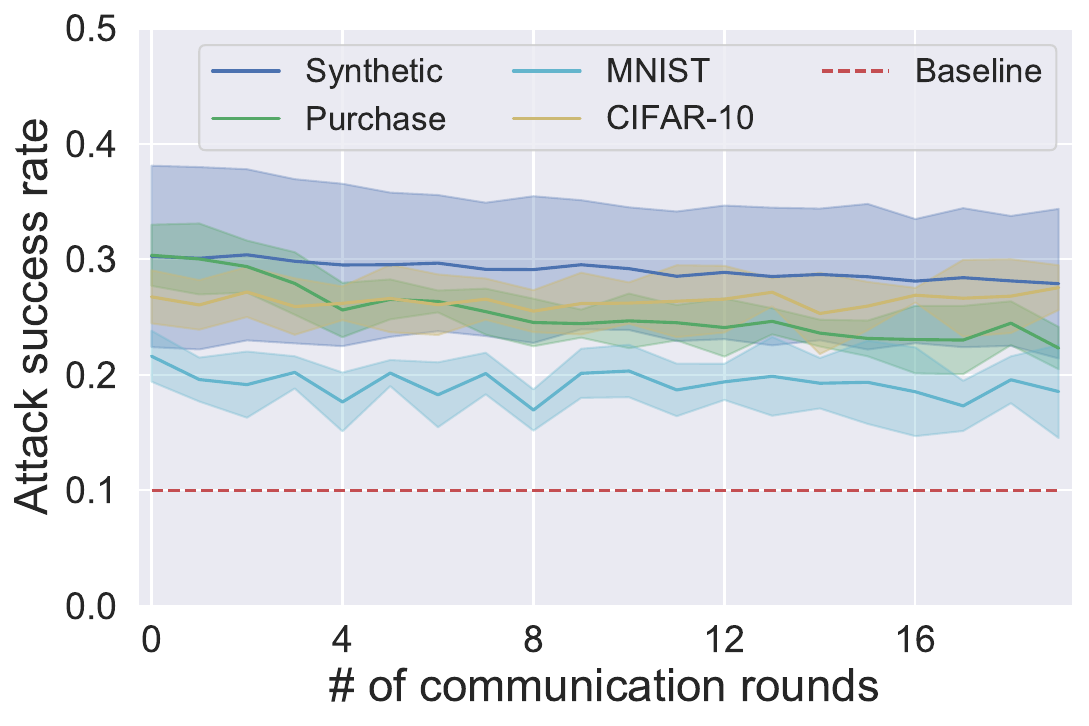}}    
%   \hspace{1.5pt}
  \subfloat[\fedavg \label{3b}]{%
        \includegraphics[width=0.33\linewidth]{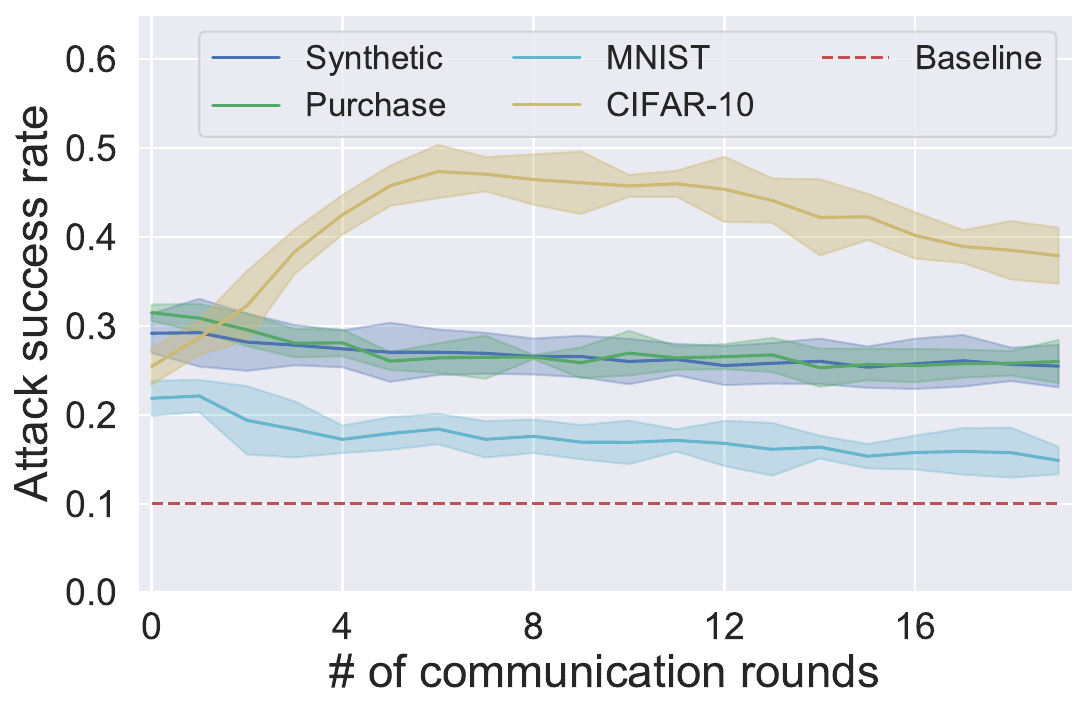}}
%   \hspace{1.5pt}
  \subfloat[\fedmd \label{3d}]{%
        \includegraphics[width=0.33\linewidth]{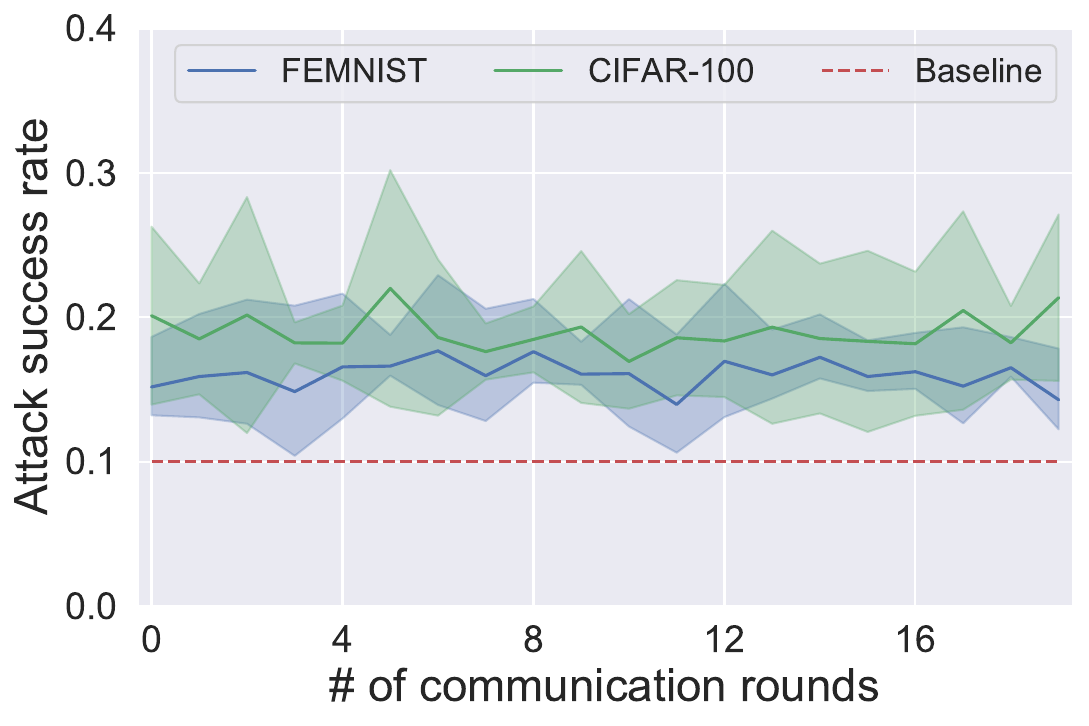}}        
  \caption{Attack success rate (ASR) of the three proposed FL frameworks during each communication round. In each plot, the $x$ axis represents the number of communication rounds, and the $y$ axis represents ASR. Each line is the mean ASR of 5 runs of experiments with shaded regions corresponding to 95\% confidence interval. (a) ASR of \fedsgd. (b) ASR of \fedavg. (c) ASR of \fedmd. As we can see, the ASR on all datasets is larger than the baseline, i.e., randomly guessing, demonstrating the effectiveness of SIAs.}
  \label{fig::asr_comm}
\end{figure*}

\subsection{Evaluation Metrics, Baseline, and Parameter Settings}

\noindent \textbf{Metric. \;} We use \textit{attack success rate} (ASR) to evaluate SIAs, which is the most commonly used metric to evaluate the performance of a given attack approach. {ASR \cite{hu2022membership} is defined as the fraction of the number of attacks that successfully identify the source of the target instances:}
\begin{equation*}
    \textrm{ASR} = \frac{\textrm{\# Successful attacks}}{\textrm{\# All attacks}}.
\end{equation*}

\noindent {For example, consider there are 100 target instances for identifying their source clients, and SIAs successfully identify 60 instances' source clients. Then, the ASR is calculated as 60\%.}

\noindent \textbf{Baseline. \;} Because we are the first to propose SIAs, there is no current work we can compare. Thus, to demonstrate the effectiveness of SIAs, we consider a trivial attack of \textit{randomly guessing} as the baseline of SIAs. Randomly guessing randomly selects a client as the source of the target training record. The ASR of randomly guessing is defined as $\frac{1}{K}$, where $K$ is the number of clients.

\noindent \textbf{Hyper-parameter setting. \;} We use the SGD optimizer for all the models with a learning rate of $0.01$. We assume there are $10$ clients in the FL. In each trial of the experiment, $100$ training records from each client are selected as the target records for SIAs. For all the learning tasks, we set the total number of communication rounds to 20, which is enough for the global model to converge. During each communication round, we record the ASR of SIAs. We report the mean of ASR over five different random seeds.

\subsection{Factors in Source Inference Attacks}

\noindent\textbf{Data distribution $\alpha$. \;} In FL, the training data across the clients are usually non-IID~\cite{li2020federated}. This means the local data from one client can not be considered as samples drawn from the overall data distribution. To simulate the heterogeneity distribution of client data, we leverage a Dirichlet distribution as in previous works~\cite{xie2019dba,bagdasaryan2020backdoor,lin2020ensemble,yurochkin2019bayesian,hsu2019measuring} to create disjoint non-IID training data for each client. The degree of non-IID is controlled by the value of $\alpha$ ($\alpha>0$) of the Dirichlet distribution. For example, $\alpha=100$ imitates almost identical local data distributions. With a smaller $\alpha$, each client is more likely to have the training records from only one class. We use Fig.~\ref{dirichlet} to visualize how the training records of CIFAR-10 are distributed among 10 clients when setting different $\alpha$ values.

\noindent\textbf{Number of local epochs $E$. \;} In \fedavg \;and \fedmd,\; each client can update their model for several epochs, and then send the model weights or predictions to the server. Many recent studies~\cite{song2017machine,carlini2019secret,murakonda2020ml,zhang2021understanding} have shown that DNN models can easily memorize their training data. Intuitively, if a client trains the local model with more epochs, the local updated model should better remember the information of the local dataset and be more confident to classify its training records. Accordingly, the prediction loss of a target record calculated from the update of the source client will be much smaller than that calculated from the updates of other clients, which will be beneficial for SIAs.

\subsection{Effectiveness of Source Inference Attacks}
To demonstrate the effectiveness of SIAs, we evaluate \fedsgd, \fedavg, and \fedmd \;under common settings: We divide the training data to the clients in the three FL frameworks with $\alpha=1$, resulting in a moderate non-IID distribution. We set $E=5$ for \fedavg, and $E_1=1, E_2=5$ for \fedmd, \;which are the common settings for FedAvg~\cite{mcmahan2017communication} and FedMD~\cite{li2019fedmd}. 

Fig.~\ref{fig::asr_comm} shows the ASR of \fedsgd, \fedavg, and \fedmd \;during each communication round. We can observe from Fig.~\ref{fig::asr_comm} that the ASR on all datasets is larger than the baseline of 10\% in each communication round, demonstrating the effectiveness of SIAs. {This indicates that a semi-honest server can steal significant source information of the training data records via our proposed SIAs in any communication rounds during federated training.} We can also see that the ASR of each FL framework differs in different datasets. This is because the local models are overfitted with a different level to their local training data. We will discuss this phenomenon in more detail in subsection~\ref{sec::why}.

\begin{mdframed}[backgroundcolor=white!10,rightline=true,leftline=true,topline=true,bottomline=true,roundcorner=2mm,everyline=true]
\textbf{Takeaway 1~}
Our proposed SIAs are effective on the FedSGD, FedAvg, and FedMD.
\end{mdframed}

\subsection{Ablation Study}
We conduct a detailed ablation study on \fedsgd, \fedavg, and \fedmd \;to learn how data distribution and local epochs influence the effectiveness of SIAs. We record and report the highest ASR during the federated training process. Table~\ref{tab::ablation_study} shows the ASR when setting different levels of non-IID data distribution and different number of local epochs.

\begin{table*}[!t]
	\caption{Understanding the impact of data distribution and local epochs in SIAs. For each value of the parameter, we report the averaged attack success rate over 5 different random seeds with its standard deviation. Because \fedsgd \;transmits gradient calculated on the current global model (equivalent to $E=1$) and does not involve training local models for several epochs, we leave the column $E=5$ and $E=10$ of the ASR of \fedsgd \;blank.}
	\label{tab::ablation_study}
	\centering
	\resizebox{1.\textwidth}{!}{%
		\begin{tabular}{lclllllllll}
			\toprule
			& 				   								& \multicolumn{9}{c}{The attack success rate (\%) of source inference attacks}                                                     \\ \cmidrule(lr){3-11}
			&  {\centering Datasets}  	& \multicolumn{3}{c}{$\alpha \!=\! 100$} 	& \multicolumn{3}{c}{$\alpha \!=\! 1$}	& \multicolumn{3}{c}{$\alpha \!=\! 0.1$}	\\ \cmidrule(lr){3-5} \cmidrule(lr){6-8} \cmidrule(lr){9-11}
			                        &      & $E \!=\! 1$ & $E \!=\! 5$ & $E \!=\! 10$ & $E \!=\! 1$ & $E \!=\! 5$ & $E \!=\! 10$ & $E \!=\! 1$ & $E \!=\!5$ & $E \!=\!10$\\ 
			                        \midrule
\multirow{4}{*}{\fedsgd}   
&  Synthetic & $ 19.1 \pm 0.4 $  & --- & --- & $30.9 \pm 2.6$ & --- & --- & $55.9 \pm 3.2$ & --- & ---\\
& Purchase & $ 15.7 \pm 0.4 $   & --- & --- & $30.6 \pm 1.0$ & --- & --- & $63.9 \pm 1.6$ & --- & --- \\ 
& MNIST & $ 12.7 \pm 0.3 $ & --- & --- & $23.1 \pm 0.5$ & --- & --- & $50.2 \pm 3.7$ & --- & ---               \\ 
& CIFAR-10 & $ 17.6 \pm  0.3$ & --- & --- & $28.5 \pm 0.7$ & --- & --- & $58.3 \pm 5.2$ & --- & ---             \\
			                 \midrule
\multirow{4}{*}{\fedavg}    
& Synthetic & $ 19.2 \pm 0.5$   & $ 19.7 \pm 0.5 $ & $ 18.9 \pm 0.6$ & $ 28.5 \pm 1.4$  & $ 28.1 \pm 1.8$ & $28.5 \pm 1.2$ & $53.6 \pm 1.3 $ & $50.8 \pm 2.6$ & $51.7 \pm 3.3$   \\
& Purchase & $15.6 \pm 0.5$ & $ 21.9 \pm 0.3$ & $ 28.2 \pm 0.5$ & $31.4 \pm 0.8$ & $32.6 \pm 0.7$ & $ 34.8 \pm 0.5$ & $67.1 \pm 0.4$ & $64.4 \pm 0.8 $ & $66.2 \pm 0.9$ \\ 
& MNIST & $ 12.1 \pm 0.1$  & $ 12.8 \pm 0.3$ & $ 13.5 \pm 0.3$ & $ 23.7 \pm 0.9$ & $23.3 \pm 0.4$ & $ 22.1 \pm 0.7$ & $58.4 \pm 4.9$ & $53.1 \pm 1.1$ & $42.3 \pm 2.6$ \\
& CIFAR-10 & $ 16.6 \pm 0.2$  & $47.8 \pm 0.4$  & $51.1 \pm 1.1$ & $ 26.3 \pm 0.7$ & $ 49.9 \pm 0.7$  & $ 55.8 \pm 0.7 $ & $56.8 \pm 3.9$ & $60.9 \pm 3.9$ & $62.5 \pm 1.9$ \\ 
                            \midrule
\multirow{2}{*}{\fedmd} 
& FEMNIST & $ 15.1 \pm 0.4$ & $17.2 \pm 0.5$ & $17.5 \pm 0.6$ & $23.2 \pm 1.1$ & $25.4 \pm 1.3$ & $24.5 \pm 1.1$ & $42.5 \pm 3.3$ & $40.6 \pm 1.0$ & $46.7 \pm 2.1$  \\
& CIFAR-100 & $18.2 \pm 0.2$ & $ 18.8 \pm 0.8$ & $20.5 \pm 0.4$ & $23.9 \pm 1.1$ & $28.1 \pm 1.9$ & $25.5 \pm 0.6$ & $40.3 \pm 1.7$ & $43.5 \pm 3.8$ & $45.6 \pm 3.9$ \\
			\bottomrule
		\end{tabular}%
	}

\end{table*}

\noindent \textbf{Evaluation of non-IID data distribution. \;} As we can see in Table~\ref{tab::ablation_study}, the ASR of all the three FL frameworks on all datasets increase when the degree of non-IID data distribution across clients increases. For example, the ASR of \fedsgd\;on CIFAR-10 increases from 17.6\% to 58.3\% when the non-IID data distribution increase from $\alpha=100$ to $\alpha=0.1$. This is because the more non-IID of the local data is, the more different the updated local models will be, which benefits SIAs. For example, for the CIFAR-10 task, a client is highly likely to have training records from only one class ({\em e.g.,} deer) when the degree of the non-IID data distribution is high. It is expected that such a client's local model will perform well in predicting its own records of deer images but perform badly in predicting the other clients' records such as dogs and trucks because the local model had never seen such images during the update process. Thus, the local model will have very small prediction losses on its own records and large losses on the other records. The distinguishable prediction losses across different local models of the clients enable the server to easily implement SIAs to infer where a training record comes from. 

\noindent \textbf{Evaluation of local epochs. \;} As shown in most scenarios in Table~\ref{tab::ablation_study}, the increase of $E$ from 1 to 10 makes the ASR of \fedsgd \;and \fedmd \;increase. For example, the ASR of \fedavg\;on CIFAR-10 increases from 16.6\% to 51.1\% when the local epochs increase from 1 to 10 when $\alpha$ is set to 100. This is because the more epochs the clients update, the more confident the local model is in predicting its training data. However, we also observe that there are scenarios, {\em e.g.,} \fedavg \;on MNIST when $\alpha=1$ and $\alpha=0.1$, that increasing $E$ does not lead to the increase of ASR but leads to a decrease. We suspect this is because training the local model with more epochs not only makes it more confident to predict its training records but also helps it to generalize better to other clients' data. In this case, the prediction losses across different local models will become less distinguishable, which leads to a decrease in ASR. We will further explain how $E$ influences ASR from the perspective of overfitting in the following section.

\begin{mdframed}[backgroundcolor=white!10,rightline=true,leftline=true,topline=true,bottomline=true,roundcorner=2mm,everyline=true]
\textbf{Takeaway 2~}
\begin{itemize}
    \item Higher data heterogeneity among local clients results in more effective SIAs.
    \item Larger local epochs in clients usually lead to more effective SIAs.
\end{itemize}
\end{mdframed}

\begin{figure*}[t!]
    \centering
  \subfloat[\fedsgd \label{4a}]{%
       \includegraphics[width=0.34\linewidth]{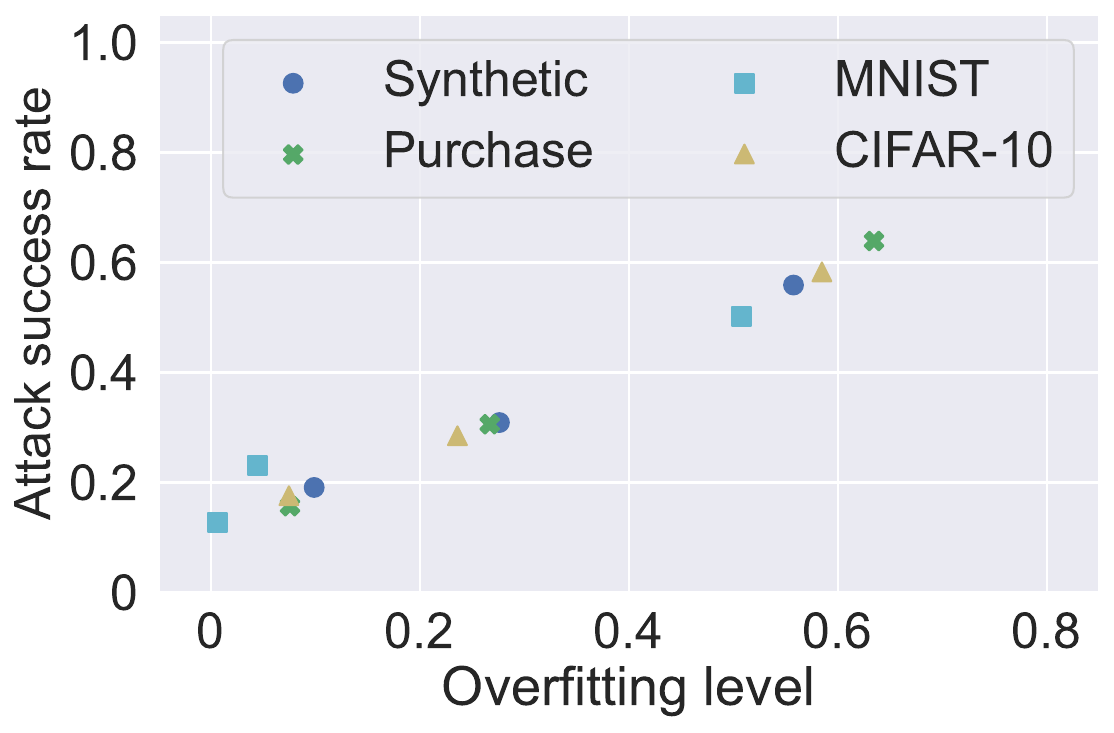}}    
%   \hspace{1.5pt}
  \subfloat[\fedavg \label{4b}]{%
        \includegraphics[width=0.34\linewidth]{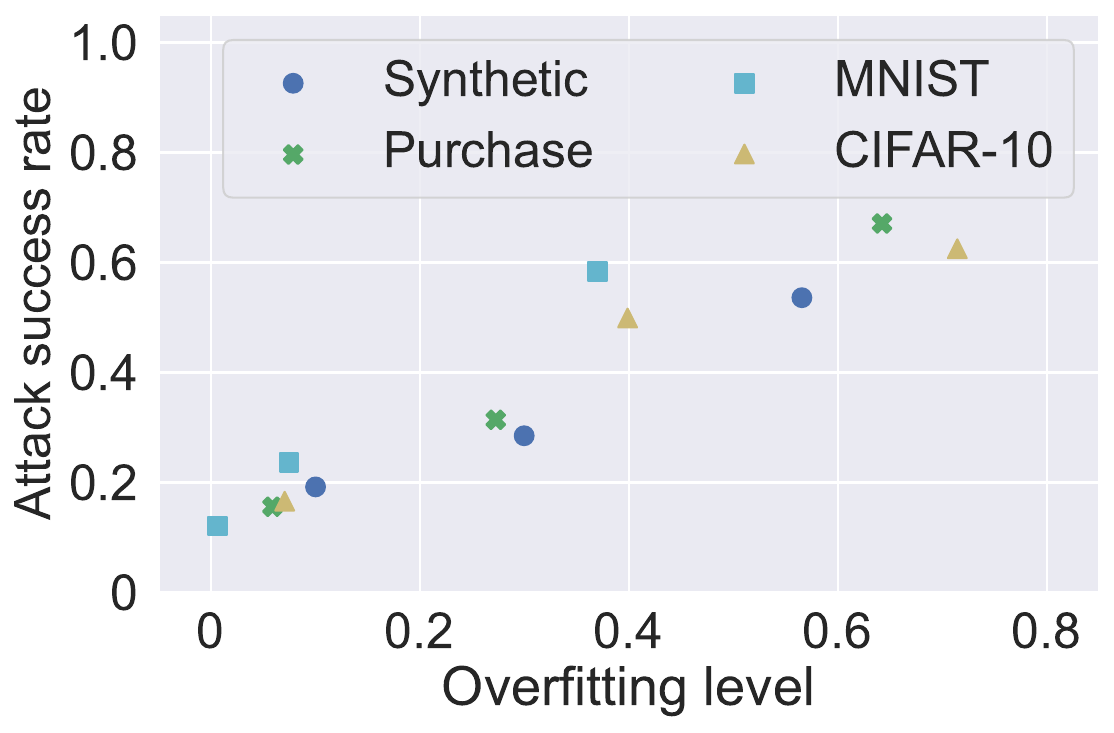}}
%   \hspace{1.5pt}
  \subfloat[\fedmd \label{4c}]{%
        \includegraphics[width=0.34\linewidth]{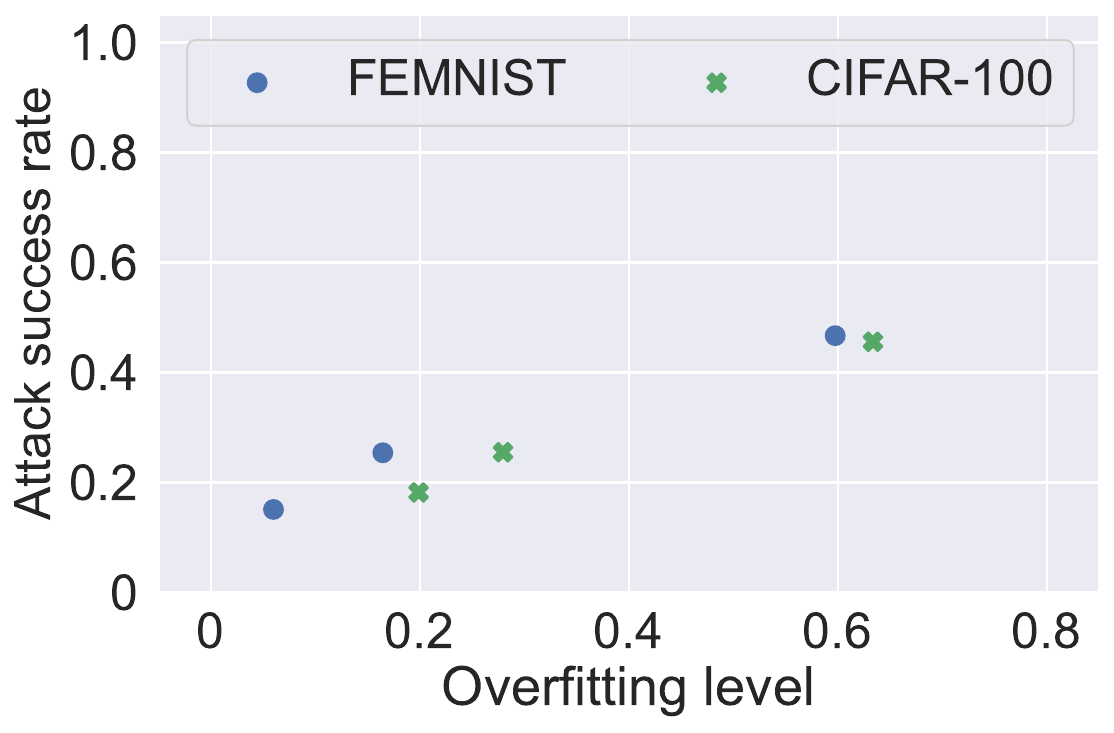}}
  \caption{Understanding the impact of overfitting in SIAs. In each plot, the $x$ axis represents the overfitting level of the FL framework, and the $y$ axis represents ASR. As we can see, for all the datasets in all the three FL frameworks, the more overfitted the local models are, the higher the ASR will be.}
  \label{fig::asr_overfit}
\end{figure*}

\begin{figure*}[t!]
    \centering
  \subfloat[\fedsgd \label{5a}]{%
       \includegraphics[width=0.34\linewidth]{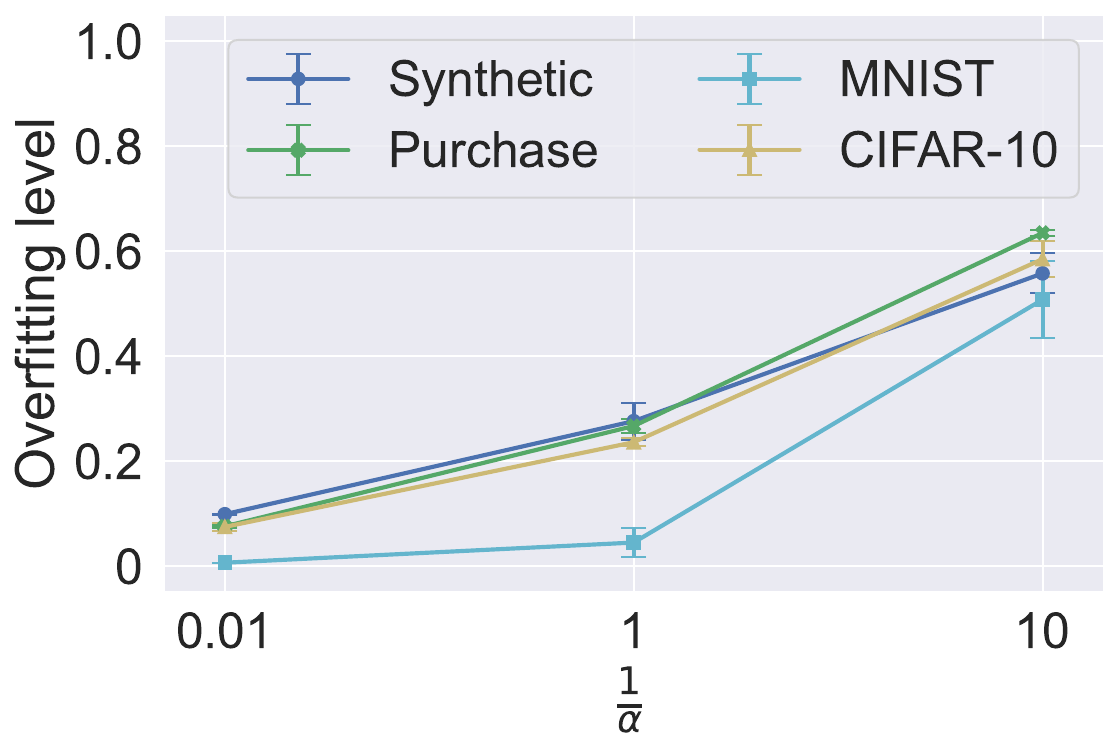}}
%   \hspace{1.5pt}
  \subfloat[\fedavg \label{5b}]{%
        \includegraphics[width=0.34\linewidth]{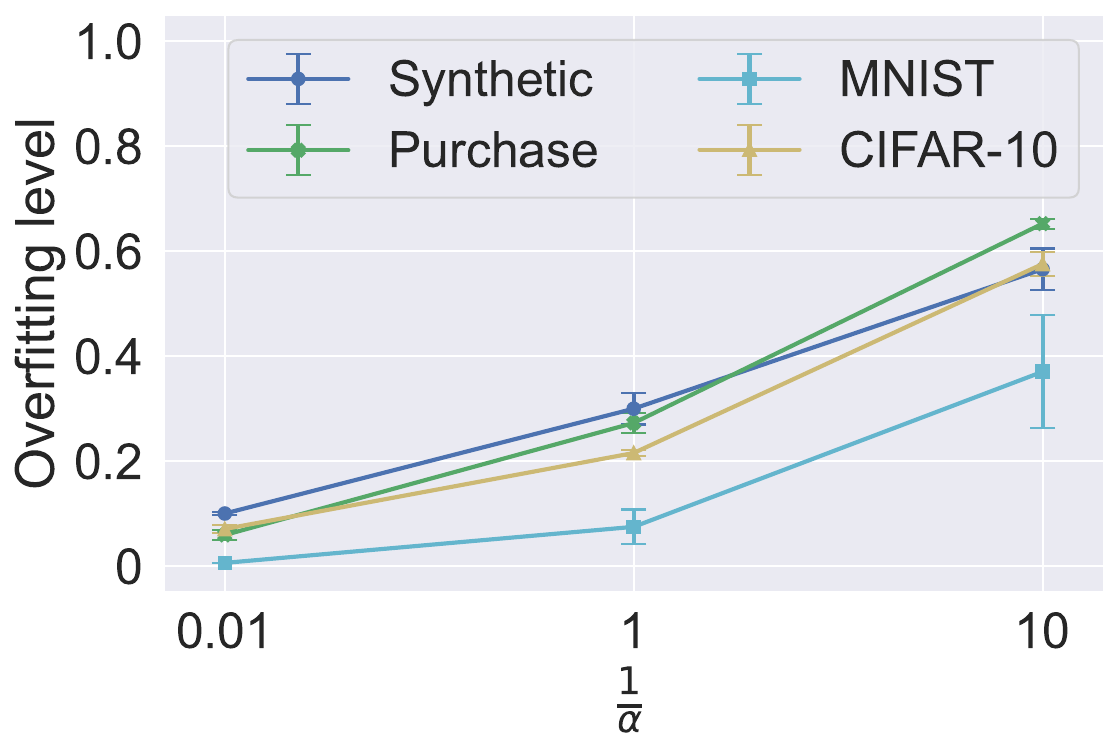}}
%   \hspace{1.5pt}
  \subfloat[\fedmd \label{5c}]{%
        \includegraphics[width=0.34\linewidth]{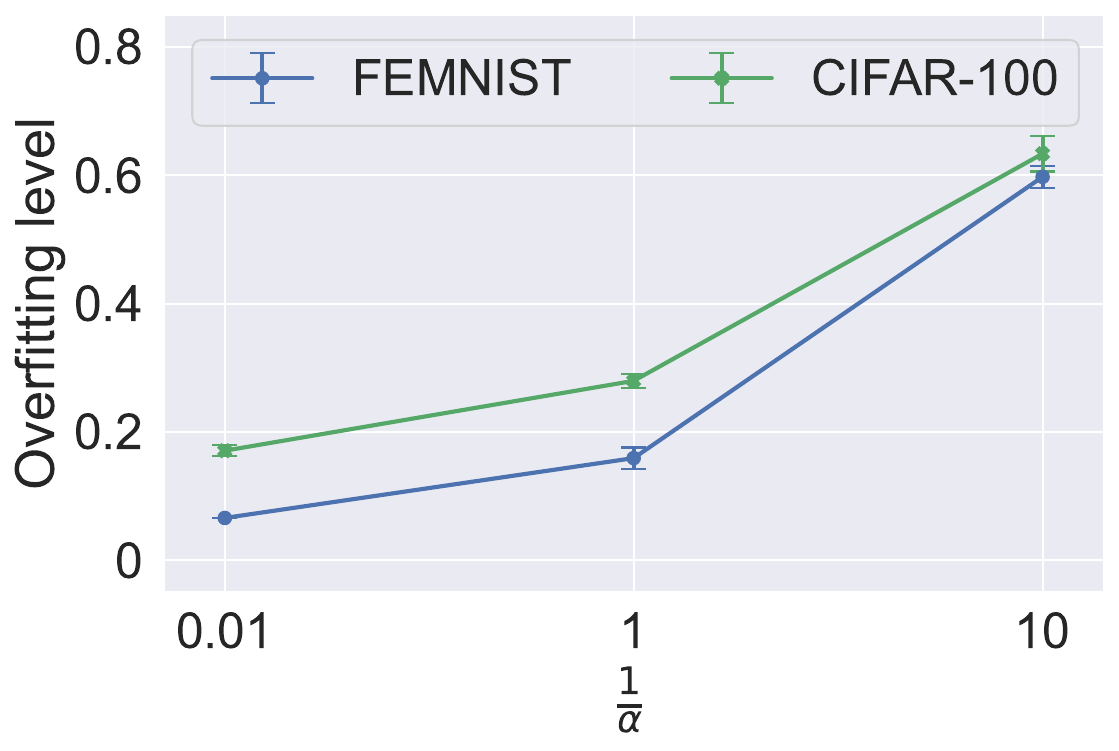}}
  \caption{Understanding the impact of the non-IID data distribution on the overfitting level of the FL framework. In the three FL frameworks, the local epoch sets to 1. In each plot, the $x$ axis represents the inverse of $\alpha$, and the $y$ axis represents the overfitting level.}
  \label{fig::overfit_noniid}
\end{figure*}

\begin{figure*}[t!]
    \centering
  \subfloat[\fedavg \label{6a}]{%
       \includegraphics[width=0.35\linewidth]{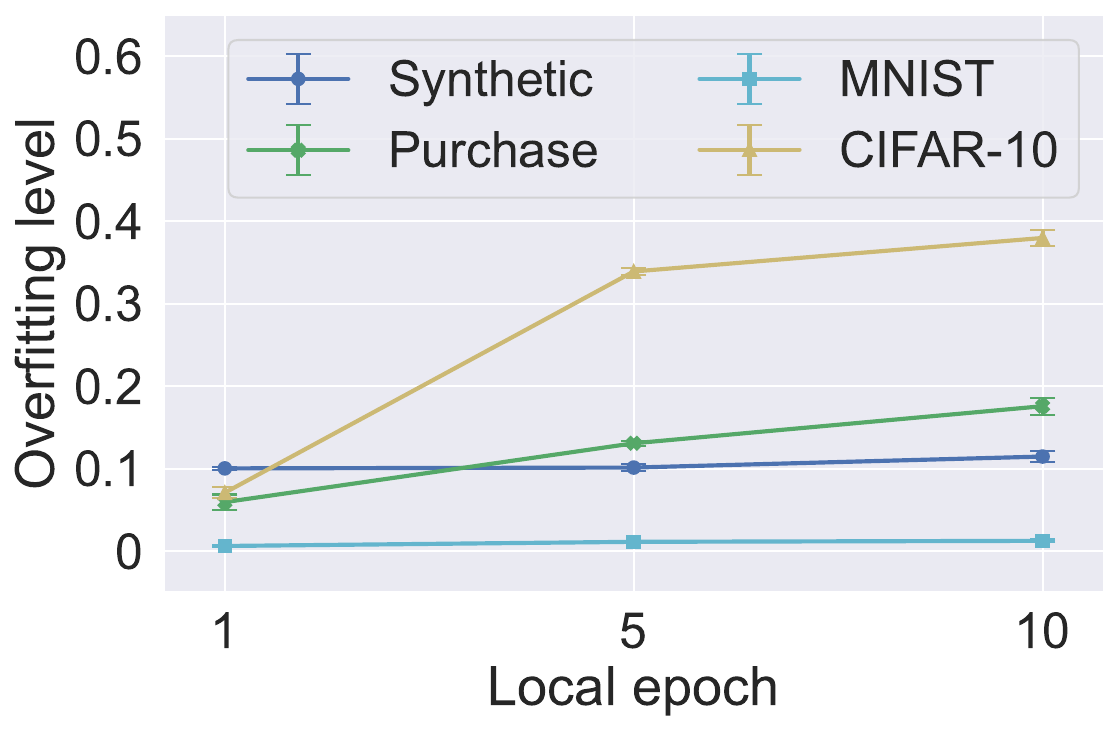}}
  \hspace{40pt}
  \subfloat[\fedmd \label{6b}]{%
        \includegraphics[width=0.35\linewidth]{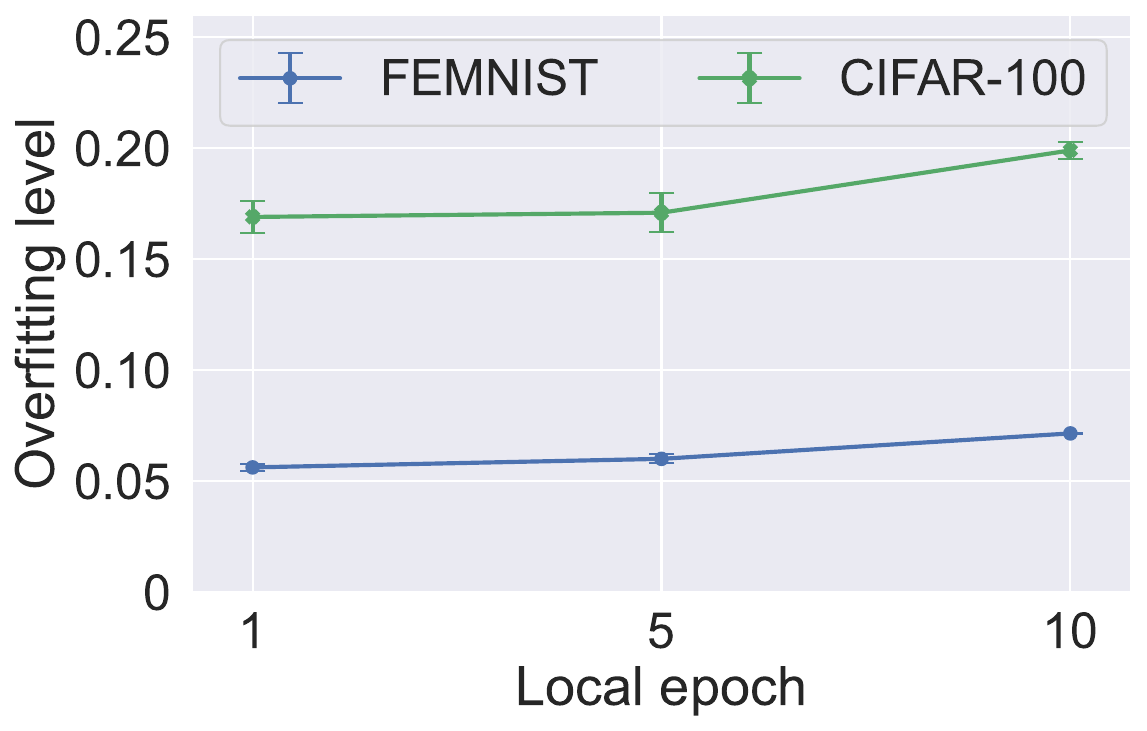}}
  \caption{Understanding the impact of the local epoch on the overfitting level of the FL framework. In the three FL frameworks, the data distribution value $\alpha$ sets to 100. In each plot, the $x$ axis represents the number of local epoch, and the $y$ axis represents the overfitting level.}
  \label{fig::overfit_epoch}
\end{figure*}

\begin{table*}[!t]
	\caption{The evaluation of source inference defenses via differential privacy. In \fedsgd \;and \fedavg, \;the training accuracy and testing accuracy is the accuracy of the trained global model. In \fedmd, \;because there is no global model, the training accuracy and testing accuracy is the averaged accuracy of the local models. The privacy budget is calculated every communication round.}
	\label{tab::defense}
	\centering
	\resizebox{1.\textwidth}{!}{%
		\begin{tabular}{lcllllllc}
			\toprule
			&  \multirow{2}[3]{*}{Datasets}  & \multicolumn{3}{c}{FL without differential privacy} 	& \multicolumn{4}{c}{FL with differential privacy}	\\ \cmidrule(lr){3-5}  \cmidrule(lr){6-9}
			                        &   & Training acc.(\%) & Testing acc.(\%) & ASR (\%) & Training acc.(\%) & Testing acc.(\%) & ASR (\%) & Privacy budgets  \\ 
			                        \midrule
\multirow{4}{*}{\fedsgd}   
&  Synthetic   & $84.3 \pm 0.5$ & $83.9 \pm 0.4$ & $30.9 \pm 2.6$ & $34.9 \pm 1.3$ & $34.7 \pm 1.1$ & $22.5 \pm 0.7$ & 1.8 \\
& Purchase     & $87.8 \pm 0.1$ & $85.9 \pm 0.2$ & $30.6 \pm 1.0$ & $14.6 \pm 0.3$ & $14.2 \pm 0.1$ & $23.2 \pm 0.6$ & 2.1  \\ 
& MNIST        & $99.3 \pm 0.1$ & $99.1 \pm 0.1$ & $23.1 \pm 0.5$ & $35.7 \pm 0.7$ & $32.1 \pm 0.5$ & $ 15.8 \pm 1.9$ & 1.9  \\ 
& CIFAR-10     & $67.8 \pm 0.1$ & $64.5 \pm 0.2$ & $28.5 \pm 0.7$ & $13.8 \pm 0.5$ & $13.7 \pm 0.3$ & $22.1 \pm 1.2$ & 1.8  \\
			                 \midrule
\multirow{4}{*}{\fedavg}    
&  Synthetic  & $94.2 \pm 0.4$ & $93.8 \pm 0.4$ & $28.5 \pm 1.2$ & $51.7 \pm 1.5$ & $51.2 \pm 0.9$ & $21.4 \pm 1.3$ & 1.9 \\
& Purchase    & $94.2 \pm 0.1 $ & $89.5 \pm 0.1$ & $34.8 \pm 0.5$ & $33.6 \pm 1.8$ & $32.1 \pm 1.3$ & $ 24.8 \pm 1.2$ & 6.8  \\ 
& MNIST       & $99.7 \pm 0.1$ & $99.2 \pm 0.1$ & $22.1 \pm 0.7$ & $10.9 \pm 0.1$ & $9.9 \pm 0.1$ & $11.8 \pm 0.5$ & 3.2               \\ 
& CIFAR-10    & $96.3 \pm 0.6$ & $70.1 \pm 0.5$ & $55.8 \pm 0.7$ & $10.1 \pm 0.1$ & $9.9 \pm 0.1$ & $ 19.5 \pm 0.6$ & 2.8             \\
                            \midrule
\multirow{2}{*}{\fedmd} 
& FEMNIST & $99.7 \pm 0.1$ & $84.1 \pm 1.5$ & $ 24.5 \pm 1.1$ & $ 16.7 \pm 0.1$ & $ 16.5 \pm 0.1$ & $16.3 \pm 0.3$ & 4.8   \\
& CIFAR-100 & $99.5 \pm 0.1$ & $72.1 \pm 0.9$ & $ 25.5 \pm 0.6$ & $ 16.6 \pm 0.1$ & $ 16.2 \pm 0.1$ & $17.1 \pm 0.5$ & 9.6 \\
			               \bottomrule
		\end{tabular}%
	}
\end{table*}

\subsection{Why Source Inference Attacks Work} \label{sec::why}
Machine Learning models especially DNNs are usually overparameterized with high complexity. This enables DNNs to learn patterns effectively from the training data on one side while on the other side such models can have unnecessarily high capacities to memorize the details of the training data~\cite{carlini2019secret,song2017machine,zhang2021understanding}, which can lead to the overfitting of ML models. An overfitted ML model cannot generalize well on its test data, {i.e.,} the model performs much better on its training data than test data. In FL, the local dataset of a client often has a limited number of records and fails to represent the whole data distribution, which exacerbates the overfitting of local models. An overfitted local model of a client is expected to have a much smaller prediction loss on a training record of its own than the loss of a training record of other clients. The distinguishable prediction losses of the local models of the different clients enable our proposed SIAs to work effectively. Moreover, the more overfitted the local models are, the more effective the SIAs will be.

We use generalization error~\cite{hardt2016train} to measure the overfitting level of the local models, which is a widely used metric to quantify overfitting in existing works~\cite{hu2022membership}. The generalization error of an ML model is defined as the absolute difference between the training accuracy and the testing accuracy of the model. Here, we first calculate the generalization error for each of the local models by using its local training dataset and the global testing dataset. Then, we calculate the averaged generalization error of the local models to reflect the overfitting level of the FL system.

\noindent \textbf{The impact of overfitting on ASR. \;} Fig.~\ref{fig::asr_overfit} shows the impact of overfitting on the performance of SIAs. As we can see, for \fedsgd, \fedavg, \;and \fedmd, the honest-but-curious server can obtain a higher attack success rate when the FL system is more overfitted to the training dataset. This observation validates our analysis of the relationship between the overfitting level and the performance of SIAs.

\noindent \textbf{The impact of non-IID distribution on overfitting. \;} Fig.~\ref{fig::overfit_noniid} examines the effect of different levels of non-IID data distribution on the overfitting level of the FL system. As we can see, for all the FL systems, increasing the level of non-IID data across the clients ({i.e.,} decreasing $\alpha$ from $100$ to $10$ and $0.1$) will inevitably increase the overfitting level. This is because the more non-IID data is, the less representative the local data will be, which makes the local model less possible to generalize well beyond its own training data.

\noindent \textbf{The impact of local epoch on overfitting. \;} Fig~\ref{fig::overfit_epoch} shows the impact of the number of epochs on the overfitting level of \fedavg \;and \fedmd. We can see that in \fedmd,\;increasing the number of local epochs does not increase the overfitting level of the FL system too much, which results in a slight increase of the attack success rate as we can observe in Table~\ref{tab::ablation_study}. In \fedavg, \;we observe that there is a large increase of the overfitting level on CIFAR-10 when we increase $E$ from 1 to 10, while there is no such trend for the other datasets. This observation explains why the ASR of CIAFR-10 is more sensitive to the local epoch compared to other datasets in Table~\ref{tab::ablation_study}.

In summary, the overfitting of the local models directly contributes to the success of SIAs, while overfitting is mainly caused by the non-IID data distribution across the clients. In FL, if the local dataset fails to represent the overall data distribution, the local model can easily overfit to the local training dataset, as depicted in Fig.~\ref{fig::overfit_noniid}. Because an overfitted local model cannot generalize well to the data beyond its local training records, it will behave differently on its training data from other clients' data, which guarantees the feasibility of SIAs. Many recent works~\cite{zhao2018federated,li2019convergence,li2020federated,li2020fedprox,lin2020ensemble} have shown that the non-IID data distribution has brought important challenges to FL such as model convergence guarantees. In this paper, we demonstrate another challenge of non-IID from the perspective of privacy: The leakage of source privacy about the training data.

{There are several promising applications where SIAs can be applied. First, as a newly proposed inference attack, SIAs can be leveraged to evaluate the privacy-preserving ability of an FL framework. A strong privacy-aware FL framework should guarantee the source privacy of the clients. Second, because no defense methods have been specifically proposed for mitigating SIAs, they can inspire researchers to propose novel defense methods or design new FL frameworks for protecting clients' source privacy. Last, we have identified that the overfitting of an FL framework is the main success factor for SIAs. Thus, SIAs can be used to help evaluate and understand the overfitting phenomenon of FL frameworks. For example, SIAs can be used to determine how many local epochs a local model should update to avoid overfitting.}

\begin{mdframed}[backgroundcolor=white!10,rightline=true,leftline=true,topline=true,bottomline=true,roundcorner=2mm,everyline=true]
\textbf{Takeaway 3~}
\begin{itemize}
    \item The overfitting of local models is the main reason why SIAs can succeed.
    \item Higher data heterogeneity results in a higher overfitting level of local models.
\end{itemize}

\end{mdframed}

\section{Discussion and Future Work}\label{sec::05}
\subsection{Defenses against SIAs}
\noindent \textbf{Differential privacy. \;} As a probabilistic privacy mechanism, differential privacy (DP)~\cite{dwork2006calibrating} provides a mathematically provable privacy guarantee. Recently, many works ~\cite{shokri2015privacy,geyer2017differentially,mcmahan2018learning,naseri2020toward,rahman2018membership} suggest DP can be applied to ML models to defend against inference attacks such as membership inference attacks and property inference attacks. When an ML model is trained with differential privacy guarantees, the learned model is expected not to learn or remember any specific data details. By definition, if the local models in FL are differentially private, the success probability of SIAs should be reduced because the communication updates calculated from such models should contain less information about the local training datasets. We discuss and evaluate whether DP can effectively mitigate SIAs.

In the experiments, we set $\alpha=1$ for all three FL frameworks and set $E=10$ for \fedavg \;and \fedmd. Here, we consider record-level DP implemented with DP-SGD~\cite{abadi2016deep}, which is the first and the most widely used differentially private training method. We train differentially private local models in each communication round before sending the updates to the central server. In our experiments, we fine-tune the noise of DP to obtain an attack success rate slightly larger than the baseline of random guess (i.e., 10\% because of 10 clients) while reporting the training and testing accuracy of the FL system.

Table~\ref{tab::defense} compares the FL systems with DP to the FL systems without DP in terms of the model utility (measured by the testing accuracy) and the attack success rate of SIAs.  As we can see, differential privacy indeed reduces the ASR of \fedsgd, \fedavg,\;and \fedmd \;on all the datasets. However, there is a significant model utility reduction in the FL systems. In some cases, the global model does not converge, {e.g.,} \fedavg\;trained on MNIST and CIFAR-10. We are aware that there exists client-level DP \cite{geyer2017differentially,mcmahan2017learning} in FL where the local clients' privacy is preserved while the utility of the FL system is maintained. However, client-level DP requires a very large number of clients (e.g., thousands in \cite{mcmahan2017learning}) to achieve the desired privacy-utility guarantee. Applying them in our experiments is not expected to achieve satisfactory results as there are only ten clients in the FL system. 

\noindent \textbf{Regularization techniques. \;} Regularization techniques such as L2-regularization and Dropout \cite{srivastava2014dropout} are leveraged to defend against membership inference attacks on machine learning models \cite{hu2022membership}. Intuitively, regularization techniques can help local models to reduce their overfitting degrees to the local training datasets. Thus, it is promising that we can leverage regularization techniques to defend against SIAs in FL. However, regularization techniques are not perfect as not all of them can achieve a satisfactory trade-off between privacy and model utility: strong regularization can also significantly reduce model performance \cite{shokri2017membership}. We leave the investigation of finding appropriate regularization techniques as a defense against SIAs for our future work.

\subsection{Limitations and Future Work for SIAs}
\noindent \textbf{Number of clients. \;} In our experiments, to demonstrate the effectiveness of SIAs, we evaluate the FL systems with ten clients, which is a small number. This FL scenario corresponds to federated to business mode where a handful of organizations jointly build a useful model, e.g., a small number of banks collaboratively build a fraud detection model \cite{lyu2022privacy,yang2019federated}. There also exists federated to customer mode where thousands or even millions of clients involved in FL systems, e.g., thousands of mobile devices jointly train a model for next-word prediction \cite{mcmahan2018learning}. Intuitively, SIAs are much more challenging under federated to customer mode because finding the source of a target record from thousands of clients is difficult. Under this setting, the performance of our proposed SIAs is expected to drop while still performing better than randomly guessing as long as the local client’s model’s behavior on its local training data is different from that of other clients (see discussion in Section \ref{sec::ana_sias}). 

\noindent \textbf{Model size.} As this is the first paper that investigates the source privacy leakage in FL, we only evaluate the FL system with relatively small deep learning models such as CNN models with only two convolutional layers but have not evaluated large models. This is because the purpose of our experiments is to show SIAs are effective in the FL frameworks of FedSGD, FedAvg, and FedMD but not to attack the best deep models. However, it would be interesting to investigate SIAs in FL with large models of millions of parameters such as VGG \cite{simonyan2014very} and Resnet \cite{he2016deep}. Models with large sizes on one side have strong learning ability while from the other side have the unnecessary capability to memorize their training data \cite{carlini2021extracting}, which might be beneficial for SIAs. We leave the investigation of how model size influences the performance of SIAs for future work.

\noindent \textbf{Best inference round. \;} Our proposed SIAs enable the server to infer the source of a target record in every communication round. As we can see in Fig. \ref{fig::asr_comm}, the attack success rate of our proposed SIAs varies in each communication round on all the datasets. This can be because the local models overfit their local training datasets to different degrees in different communication rounds. There are two promising directions after this paper: i) It would be interesting if we can find the communication round that can achieve the best attack performance; ii) Since the server can save the updates of the clients in every communication round, it would be interesting to investigate the possibility of proposing a new attack approach that can leverage all the information collected during training to achieve better performance.

\section{Related work}\label{sec::06}

\subsection{Privacy Attacks in FL}
We summarize the existing privacy attacks and compare them with our proposed source inference attacks in FL in Table \ref{tab:attack_comparison}. For each of the existing attacks, we introduce them with more detailed descriptions as follows.
\begin{table*}[t]
\centering
\caption{A summary of privacy attacks in federated learning.}
\label{tab:attack_comparison}
\resizebox{\linewidth}{!}{

\begin{tabular}{@{}lccccl@{}}
\toprule
\multirow{2}{*}{\begin{tabular}[c]{@{}l@{}}\textbf{Privacy}  \textbf{Attacks}\end{tabular}} & \multicolumn{2}{c}{\textbf{Attacker}} & \multicolumn{2}{c}{\textbf{FL framework}} & \multicolumn{1}{c}{\multirow{2}{*}{\begin{tabular}[c]{@{}c@{}}\textbf{Attack Goal}\end{tabular}}} \\ 
\cmidrule(l){2-3} \cmidrule(l){4-5}  
 & \begin{tabular}[c]{@{}c@{}} Client \end{tabular} & \begin{tabular}[c]{@{}c@{}} Server \end{tabular} &  \begin{tabular}[c]{@{}c@{}} Horizontal FL \end{tabular} & 
 \begin{tabular}[c]{@{}c@{}} Vertical FL \end{tabular}
 &  \begin{tabular}[c]{@{}c@{}}\end{tabular}\\
\midrule
Data reconstruction attacks~\cite{boenisch2021curious} & \pie{0} & \pie{360} & \pie{360} & \pie{0} & Reconstruct training data of the client\\
Property inference attacks~\cite{melis2019exploiting} & \pie{360} & \pie{0} & \pie{360} & \pie{0}  & Infer property of other clients' training data \\
Feature inference attacks \cite{luo2021feature} & \pie{360} & \pie{0} & \pie{0} & \pie{360} & Infer features in other clients' training data \\
Preference profiling attacks \cite{zhou2022ppa} & \pie{0} & \pie{360} & \pie{360} & \pie{0} & Infer the training data distribution of the client \\
Membership inference attacks \cite{nasr2019comprehensive} & \pie{360} & \pie{360} & \pie{360} & \pie{0} & Infer the membership information of a data sample\\
\midrule
Source inference attacks. (\textbf{Ours})& \pie{0} & \pie{360} & \pie{360} & \pie{0}  & Infer the source information of a training sample \\ 

\bottomrule
\end{tabular}
}

\begin{tablenotes}
\item[] \pie{360}: applicable; \pie{0}: not applicable
\end{tablenotes}

\end{table*}

\noindent \textbf{Data reconstruction attacks. \;} This type of attack aims to reconstruct individual client's class-wise training data records that represent a whole class or instance-wise training data records. Class-wise data reconstruction attacks are usually achieved by GANs techniques \cite{goodfellow2014generative} that leverage the model updates as discriminators to generate generic representations of class-wise data records \cite{hitaj2017deep,wang2019beyond}. Instance-wise data reconstruction attacks leverage optimization techniques to iteratively optimize a dummy data record with a label so that the gradients on the dummy record are close to the gradients uploaded from a client \cite{zhu2019deep,zhao2020idlg}. The instance-wise data reconstruction attacks firstly were limited in a setting where the mini-batch is one \cite{zhu2019deep,zhao2020idlg}, and recent works \cite{yin2021see,boenisch2021curious} have relaxed this assumption and demonstrate the effectiveness in the setting of large mini-batch sizes.

\noindent \textbf{Property inference attacks. \;} This type of attack aims to infer the properties of clients' training data, {e.g.,} inferring when a particular person first appears in a client’s photos or when the client begins to visit a certain type of location~\cite{melis2019exploiting}. Property inference attacks are usually achieved by training a binary classifier that takes as input the parameters of the model and outputs whether the model's training data has the target property or not \cite{wang2022poisoning,ganju2018property}. 

\noindent {\textbf{Feature inference attacks. \;} This type of attack targets vertical FL (see Section \ref{sec::pre_fl} for a detailed introduction of vertical FL) and aim to infer the feature values of clients \cite{luo2021feature}. Feature inference attacks leverage the solving of mathematical equality for simple models such as logistic regression because the prediction output and the input containing the target features can be constructed as a set of equations. For complex models such as neural networks, a generative regression network is trained through an optimization process and is leveraged to compute the target features. Recently, the work \cite{luo2022feature} demonstrates that feature inference attacks can also reconstruct the private input on centralized trained models based their explanation values of predictions. The work \cite{luo2022feature} shows that an adversary having an auxiliary dataset and black-box access to the model can successfully attack popular Machine Learning as a Service platforms such as Google Cloud and IBM aix360.}

\noindent \textbf{Preference profiling attacks. \;} This type of attack aims to infer the data preference of clients, e.g., in the FL application of recommender systems, a malicious server infers which item a client likes or dislikes \cite{zhou2022ppa}. The attack intuition is the sensitivity of gradients during the FL training reflects the sample size of a class: a small gradient change indicates the sample size of a class is large. The statistical heterogeneity in FL amplifies the gradient sensitivity across classes and thus benefits the preference profiling attacks.

\noindent \textbf{Membership inference attacks. \;} {Membership inference attacks (MIAs), which are the most related inference attacks to the proposed SIAs, aim to identify the training data of a model \cite{hu2022membership}. MIAs are usually achieved by training a binary classifier \cite{shokri2017membership} or leveraging a threshold of a metric such as prediction confidence \cite{salem2019ml} and prediction entropy \cite{song2021systematic} to decide whether a data record is training data or not. Recent studies \cite{chobola2022membership,wu2022membership,conti2022label,hu2022m,wang2022debiasing,zhang2022label} have shown that many different types of model such as semantic segmentation models, text-to-image generation models, graph neural networks, multi-modal models, and recommender systems are vulnerable to MIAs. Currently, most studies of MIAs \cite{shokri2017membership,yeom2018privacy,salem2019ml,he2022semi,yuan2022membership,hu2022membershipbackdoor} are investigated under centralized settings where one dataset containing all training instances is used for training the models.}

{In the context of FL, MIAs were first investigated in FedSGD where a malicious client or server can infer whether or not a specific location profile was used for federated training based on the observation of the non-zero gradients of the embedding layer of the global model \cite{melis2019exploiting}. Then, MIAs are investigated in FedAvg where a malicious client can passively infer the membership privacy of the FL system or actively craft her updated model parameters to steal more membership privacy \cite{nasr2019comprehensive}. However, because the purpose of MIAs is to distinguish training data from testing data of the FL system, the existing research of MIAs ignores exploring the source client of the training data. In this paper, we propose SIAs to fill this gap and demonstrate SIAs are effective in different FL frameworks.}

\subsection{Privacy Defenses in FL}
\noindent \textbf{Cryptography-based defense.
\;} Cryptography-based defense leverages homomorphic encryption (HE) or secure multiparty computation (SMC) techniques to defend against privacy leakage in FL. In HE-based FL systems \cite{zhang2020batchcrypt,cheng2021secureboost,hardy2017private,liu2019secure,liu2020secure,nikolaenko2013privacy,zheng2022aggregation,jebreel2022enhanced}, each local client encrypts their gradients or model parameters using a public key and sends
the ciphertexts to a central server. In SMC-based FL systems \cite{bonawitz2017practical,mohassel2017secureml}, each client adds random values to their gradients or model parameters for masking their true updates. The cryptography-based FL systems prevent the server from knowing the clients' local updates. Under this type of defense, the existing privacy attacks like membership inference attacks \cite{nasr2019comprehensive}, property inference attacks \cite{melis2019exploiting}, data reconstruction attacks \cite{zhu2019deep,zhao2020idlg,boenisch2021curious}, preference profiling attack \cite{zhou2022ppa}, as well as our proposed SIAs cannot be immediately mounted. However, HE and SMC are computationally expensive that increase substantial additional communication and computation overheads in the FL system, which may prevent clients with limited computing resources and bandwidth from participating in FL. \cite{lyu2022privacy}.

\noindent \textbf{Differential privacy. \;}
{Compared to cryptography-based methods, DP is more computationally efficient, adding noise directly to the private data to trade off privacy and utility~\cite{lyu2022privacy}. Existing works of privacy-preserving FL mainly focus on either the centralized DP mechanism that requires a trusted central server~\cite{geyer2017differentially,mcmahan2018learning} to add noise to the local updates or a local differential privacy mechanism where each client perturbs its updates before sending it to an untrusted aggregator~\cite{truex2019hybrid,sun2020ldp}. These privacy-preserving methods ~\cite{geyer2017differentially,mcmahan2018learning,bonawitz2017practical,li2019fedmd,sun2020ldp,liu2020secure} have been evaluated and demonstrated to be effective for mitigating privacy attacks in FL. However, because of the perturbations caused by DP noise, the FL systems inevitably sacrifice the utilities of the models for mitigating privacy attacks \cite{wei2020federated,yuan2022membership}.}

{In this paper, we investigate whether the most widely used DP mechanism DP-SGD \cite{abadi2016deep} can help to mitigate SIAs in FL. The intuition is that the semi-honest server will be difficult to infer private information of a single instance from the differentially private local models because they do not contain details of the local datasets. However, as demonstrated in our experiments, applying vanilla DP in FL does not produce satisfactory results for mitigating SIAs, since it suffers from an unacceptable trade-off between model accuracy and defense effectiveness against SIAs. Recently, several DP mechanisms \cite{yuan2022membership,wei2021user,zheng2022balancing} have been designed particularly for FL training that can achieve better privacy-utility trade-offs than DP-SGD. For such advanced DP mechanisms in FL, it is worth investigating whether they can effectively mitigate our proposed SIAs while maintaining the high utility of the FL model.}

\section{Conclusion}\label{sec::07}
In this paper, we propose a new inference attack called SIAs in federated learning, which allows an honest-but-curious server to identify the source client of a training record. We adopt the Bayesian theorem to derive an inference method enabling the server to gain significant source information of the training data during each communication round. We propose \fedsgd, \fedavg,\;and \fedmd\;showing that in existing FL frameworks, the clients sharing gradients, model parameters, or the predictions on a public dataset will lead to source information leakage to the server. We evaluate SIAs on many datasets under different federated settings. The comprehensive experimental results validate the effectiveness of SIAs. We also conduct a detailed ablation study to investigate how the non-IID data distribution and local epochs impact the performance of SIAs. We identify that the overfitting of local models is the key factor contributing to the success of SIAs. We evaluate differential privacy as a defense mechanism to mitigate SIAs, but the experimental results suggest differential privacy is not a good solution because of the unacceptable trade-offs between model utilities and the defense effectiveness against SIAs. We discuss the limitations of our proposed SIAs and potential research opportunities. The investigation of finding appropriate regularization techniques as a
defense against SIAs to achieve satisfactory privacy-utility trade-offs becomes our future work.

\section*{Acknowledgement}
Dr. Xuyun Zhang is supported only by ARC DECRA Grant DE210101458. The work of K.-K. R. Choo was supported only by the Cloud Technology Endowed Professorship.

\ifCLASSOPTIONcaptionsoff
  \newpage
\fi

\bibliographystyle{IEEEtran}
\bibliography{reference_long.bib}

\begin{IEEEbiography}[{\includegraphics[width=1in,height=1.25in,clip,keepaspectratio]{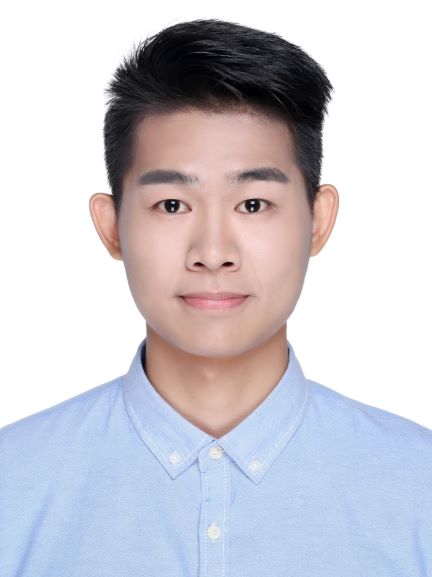}}]{Hongsheng Hu}
is currently a PhD at Faculty of Engineering, University of Auckland, New Zealand. His research focuses on AI privacy and security, especially membership inference attacks, differential privacy, and inference attacks in the context of federated learning. He has published 8 international refereed journal and conference papers, including ACM Computing Surveys, IJCAI, and ICDM.
\end{IEEEbiography}

\begin{IEEEbiography}[{\includegraphics[width=1in,height=1.25in,clip,keepaspectratio]{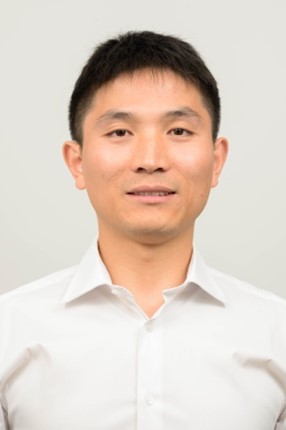}}]{Xuyun Zhang}
is currently working as a senior lecturer in School of Computing at Macquarie University (Sydney, Australia). Besides, he has the working experience in University of Auckland and NICTA (now Data61, CSIRO). He received his PhD degree in Computer and Information Science from University of Technology Sydney (UTS) in 2014, and his MEng and BSc degrees from Nanjing University. His research interests include scalable and secure machine learning, big data mining and analytics, big data privacy and cyber security, cloud/edge/service computing and IoT, etc. He is the recipient of 2021 ARC DECRA Award and several other prestigious awards, and has been listed as one of the Clarivate 2021 Highly Cited Researchers.
\end{IEEEbiography}

\begin{IEEEbiography}[{\includegraphics[width=1in,height=1.25in,clip,keepaspectratio]{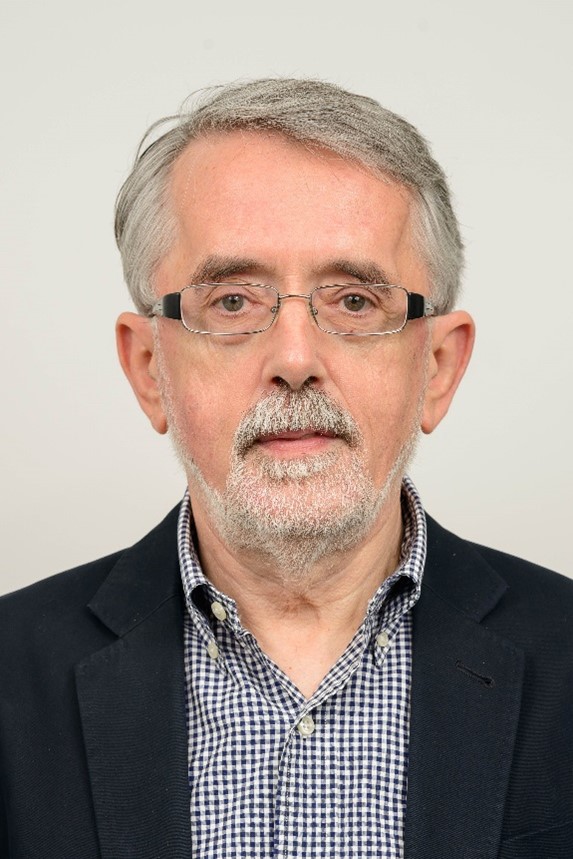}}]{Zoran Salcic}
(Life Senior Member, IEEE) received the B.E., M.E., and Ph.D. degrees in electrical and computer engineering from Sarajevo University in 1972, 1974, and 1976, respectively. He is a Professor and the Chair of computer systems engineering with University of Auckland, New Zealand. He has published more than 400 peer-reviewed journal and conference papers, and several books. His main research interests include various aspects of cyber-physical systems that include complex digital systems design, custom-computing machines, design automation tools, hardware–software co-design, formal models of computation, sensor networks and Internet of Things, and languages for concurrent and distributed systems and their applications such as industrial automation, intelligent buildings and environments, and collaborative systems with service robotics, IoT, big data processing and many more. He is a Fellow of the Royal Society of New Zealand. He was a recipient of the Alexander von Humboldt Research Award in 2010.
\end{IEEEbiography}

\begin{IEEEbiography}[{\includegraphics[width=1in,height=1.25in,clip,keepaspectratio]{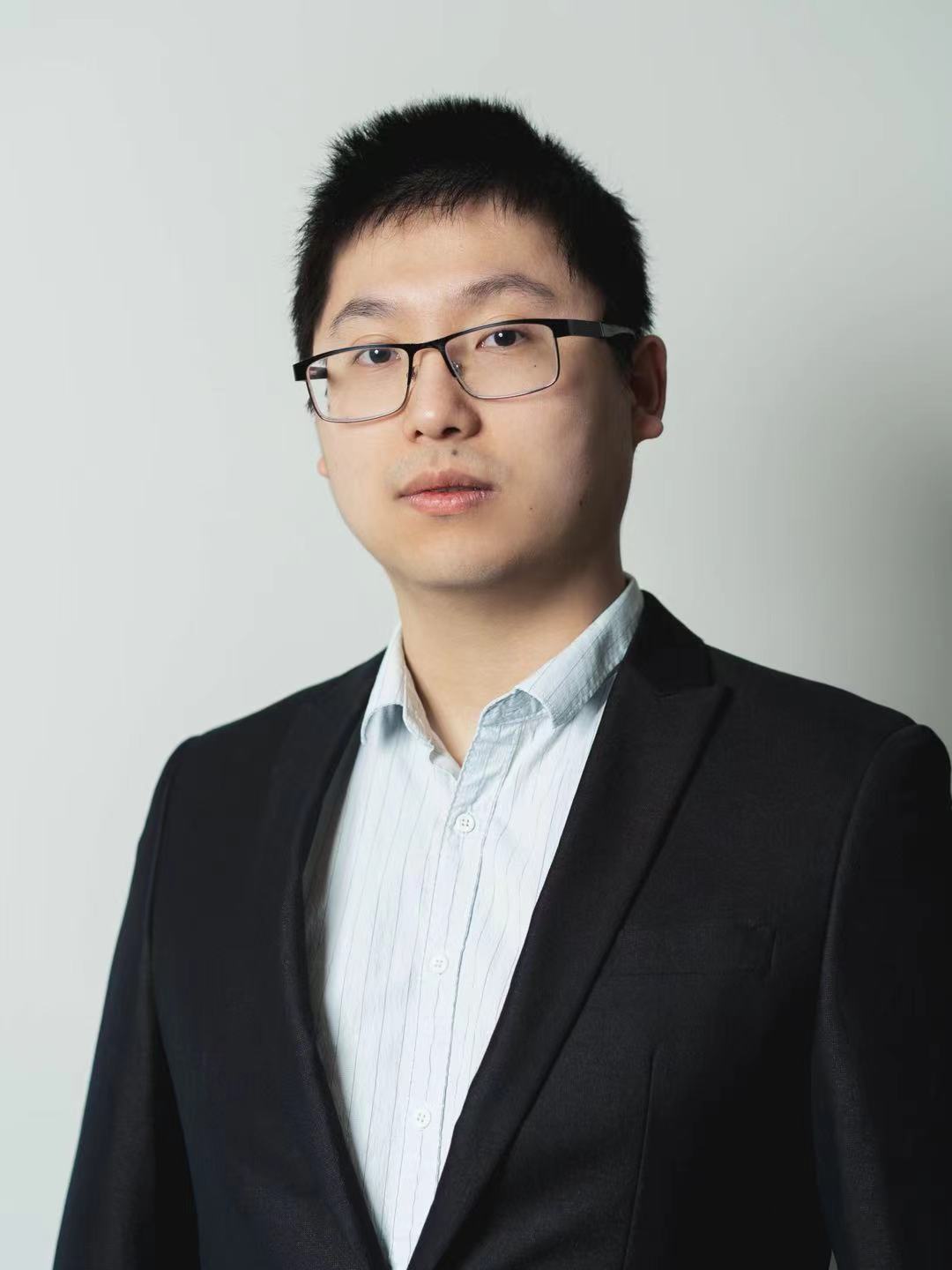}}]{Lichao Sun}
is currently an Assistant Professor in the Department of Computer Science and Engineering at Lehigh University. Before that, he received my Ph.D. degree in Computer Science at University of Illinois, Chicago in 2020, under the supervision of Prof. Philip S. Yu. Further before, he obtained M.S. and B.S. from University of Nebraska Lincoln. His research interests include security and privacy in deep learning and data mining. He mainly focuses on AI security and privacy, social networks, and natural language processing applications. He has published more than 45 research articles in top conferences and journals like CCS, USENIX-Security, NeurIPS, KDD, ICLR, AAAI, IJCAI, ACL, NAACL, TII, TNNLS, TMC.
\end{IEEEbiography}

\begin{IEEEbiography}[{\includegraphics[width=1.0in,height=1.25in,clip,keepaspectratio]{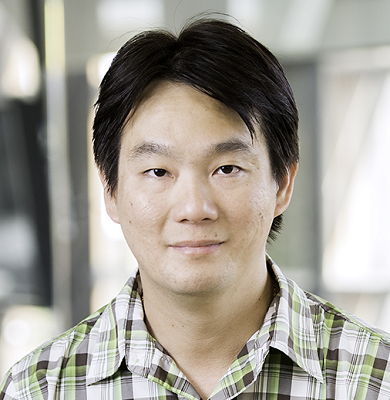}}]{Kim-Kwang Raymond Choo} (Senior Member, IEEE) received the Ph.D. degree in information security from the Queensland University of Technology, Australia, in 2006. He currently holds the Cloud Technology Endowed Professorship at The University of Texas at San Antonio. He was a recipient of the 2022 IEEE Hyper-Intelligence Technical Committee (HITC) Award for Excellence in Hyper-Intelligence (Technical Achievement), the 2022 IEEE Technical Committee on Homeland Security (TCHS) Research and Innovation Award in 2022, the 2022 IEEE Technical Committee on Secure and Dependable Measurement (TCSDM) Mid-Career Award, and the 2019 IEEE Technical Committee on Scalable Computing (TCSC) Award for Excellence in Scalable Computing (Middle Career Researcher).
\end{IEEEbiography}

\begin{IEEEbiography}[{\includegraphics[width=1in,height=1.25in,clip,keepaspectratio]{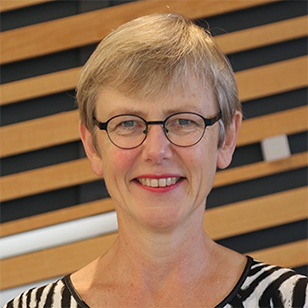}}]{Gillian Dobbie}
obtained a PhD in Computer Science from the University of Melbourne. She joined the Department of Computer Science at the University of Auckland in 2001, and is now a Professor. Gill’s primary research area covers  the modelling, management and efficient processing of data. She has worked on a diverse range of topics including formal methods, logical foundations of data models, modeling of semistructured data, data warehousing, data privacy, data mining, and recurrent pattern mining. Her most recent research projects have focused on adversarial attacks and defences. She has published over 160 international refereed journal and conference papers. 
\end{IEEEbiography}
% % insert where needed to balance the two columns on the last page with
% %\newpage


\end{document}